\documentclass[leqno,11pt,letter]{amsart}

\usepackage[full]{textcomp}
\usepackage{newpxtext}
\usepackage{times}
\usepackage{cabin} 
\usepackage[varqu,varl]{inconsolata} 
\usepackage[bigdelims,vvarbb]{newpxmath} 
\usepackage[cal=cm,bb=esstix,bbscaled=1.05]{mathalfa} 
\usepackage[bottom=1in,top=1in,left=1.25in,right=1.25in]{geometry}
\usepackage{savesym}
\let\savedegree\bigtimes
\let\bigtimes\relax
\usepackage{mathabx}
\let\bigtimes\savedegree
\usepackage[usenames,dvipsnames]{color}
\usepackage{hyperref}
\hypersetup{
 colorlinks,
 linkcolor={blue!90!black},
 citecolor={red!80!black},
 urlcolor={blue!50!black}
}
\usepackage{tikz}
\usetikzlibrary{decorations.markings}
\usepackage{mathrsfs}
\usepackage[all,cmtip]{xy}
\usepackage{mleftright}
\usepackage{booktabs}

\usepackage{cite}
\usepackage{enumitem}

\usepackage{graphicx}
\usepackage{subfigure}
\usepackage{caption}
\usepackage{placeins}
\usepackage{comment}
\usepackage[normalem]{ulem} 

\newcommand{\Hilbert}{H}
\newcommand{\Projection}{P}

\newcommand{\la}{\langle}
\newcommand{\ra}{\rangle}
\newcommand{\sgn}{\operatorname{sgn}}
\newcommand{\bds}{\boldsymbol}
\newcommand{\PV}{\operatorname{PV}}
\newcommand{\ph}{\text{phys}}

\setlist[enumerate]{labelsep=*, leftmargin=1.5pc}
\setlist[enumerate]{label=\normalfont(\roman*), ref=\roman*}

\newtheorem{theorem}{Theorem}[section]

\newtheorem{cor}[theorem]{Corollary}

\theoremstyle{definition}

\numberwithin{equation}{section}

\newcommand{\IGNORE}[1]{}
\newcommand{\ignore}[1]{}
\newcommand{\veps}{\varepsilon}

\newcommand{\re}{\operatorname{Re}}
\newcommand{\im}{\operatorname{Im}}

\newcommand{\phm}{\phantom{-}}
\newcommand{\mbb}[1]{\mathbb{#1}}

\newcommand{\mc}[1]{\mathcal{#1}}

\newcommand{\jd}{\displaystyle}

\newcommand{\der}[2]{\frac{\partial #1}{\partial #2}}

\newcommand{\black}[1]{\textcolor{black}{#1}}
\newcommand{\pa}{\partial}
\newcommand{\wtil}{\widetilde}
\newcommand{\e}[1]{{(#1)}}

\usepackage{xcolor}
\usepackage{etoolbox}
\makeatletter
\pretocmd\@bibitem{\color{black}\csname keycolor#1\endcsname}{}{\fail}
\newcommand\citecolor[1]{\@namedef{keycolor#1}{\color{blue}}}
\makeatother

\begin{document}

\author[Jon Wilkening and Xinyu Zhao]{Jon Wilkening and Xinyu Zhao}
\address{Department of Mathematics\\ University of California at
  Berkeley\\Berkeley, CA\\94720\\USA}
\email{wilkening@berkeley.edu}
\email{zhaoxinyu@berkeley.edu}
\thanks{This work was supported in part by the National Science
  Foundation under award number DMS-1716560 and by the Department of
  Energy, Office of Science, Applied Scientific Computing Research,
  under award number DE-AC02-05CH11231.}
\subjclass[]{}
\title{Quasi-periodic traveling gravity-capillary waves}

\begin{abstract}
  We present a numerical study of spatially quasi-periodic traveling
  waves on the surface of an ideal fluid of infinite depth. This is a
  generalization of the classic Wilton ripple problem to the case when
  the ratio of wave numbers satisfying the dispersion relation is
  irrational. We propose a conformal mapping formulation of the water
  wave equations that employs a quasi-periodic variant of the Hilbert
  transform to compute the normal velocity of the fluid from its
  velocity potential on the free surface. We develop a Fourier
  pseudo-spectral discretization of the traveling water wave equations
  in which one-dimensional quasi-periodic functions are represented
  by two-dimensional periodic functions on the torus. This leads to
  an overdetermined nonlinear least squares problem that we solve
  using a variant of the Levenberg-Marquardt method. We investigate
  various properties of quasi-periodic traveling waves, including
  Fourier resonances, time evolution in conformal space on
    the torus, asymmetric wave crests, capillary wave patterns
  that change from one gravity wave trough to the next without repeating,
  and the dependence of wave speed and surface
  tension on the amplitude parameters that describe a two-parameter
  family of waves.
\end{abstract}


\maketitle
\markboth{J. WILKENING AND X. ZHAO}{QUASI-PERIODIC TRAVELING WATER WAVES}


\section{Introduction}
\label{sec:intro}


Traveling water waves have long played a central role in the field of
fluid mechanics. Following a tradition dating back to Stokes
  \cite{stokes1847,craik:05}, most work on traveling waves has assumed
  periodic boundary conditions; see e.g.~\cite{lamb:hydro,
    milne:thomson:hydro, johnson97, nekrasov1921steady,
    levi1925determination, beale1979existence, toland1985bifurcation,
    jones1989symmetry}.  Solitary water waves that propagate on the
  real line but decay to zero at infinity also have a long history
  \cite{rayleigh:1876} and have been studied extensively
  \cite{friedrichs:57, bslh:76, vandenBroeck:92,
    amick:81, milewski:10, vandenBroeck:book}. A third option is to
  assume spatially quasi-periodic boundary conditions.  These arise
  naturally in many contexts related to water waves, which we briefly
  outline below. However, to date, spatially quasi-periodic water
  waves have only been investigated through weakly nonlinear models
  \cite{bridges1996spatially, zakharov1968stability,
    janssen2003nonlinear, ablowitz2015interacting} or through a
  Fourier-Bloch stability analysis in which the eigenfunctions of the
  linearization about a Stokes wave have a different period than the
  Stokes wave \cite{longuet:78, oliveras:11,
    trichtchenko:16}. No methods currently exist to study the
  long-time evolution of unstable subharmonic perturbations under the
  full water wave equations nor to compute quasi-periodic traveling
  waves beyond the weakly nonlinear regime.  Our goal in this paper
  and its companion \cite{quasi:ivp} is to address this gap and
  develop a mathematical and computational conformal mapping framework
  to study fully nonlinear spatially quasi-periodic water waves,
  focusing here on traveling waves and in \cite{quasi:ivp} on the
  time-dependent initial value problem.

In oceanography, modulational instabilities of periodic narrowband
wavetrains are thought to contribute to the formation of rogue waves
in the open ocean \cite{osborne2000, janssen2003nonlinear}. The
nonlinear dynamics are usually approximated by the nonlinear
Schr\"odinger equation \cite{benney:newell:67, zakharov1968stability}
and the growth of unstable modes is governed by the Benjamin-Feir
instability \cite{benj:feir:67}. Three-dimensional effects of
multi-phase interacting wave trains are also believed to be important
in the growth of unstable modes and rogue wave generation
\cite{bridges2005, onorato2006modulational,
  ablowitz2015interacting}. Along these lines, an interesting open
question is whether wave trains of different wavelength co-propagating
in the same direction might have interesting stability properties. We
present in this paper a method of computing spatially quasi-periodic
traveling wave trains of this type on deep water, leaving the
stability question for future research.

Modulational instabilities of periodic wavetrains bring in unstable
modes that grow exponentially until nonlinear effects become
important. As noted by Osborne et.~al.~\cite{osborne2000}, one expects
Fermi-Pasta-Ulam recurrence in this scenario \cite{berman:FPU:2005}. An example
of such recurrence in the context of standing waves was given by
Bryant and Stiassnie \cite{bryant:stiassnie:94} when the wavelength of the
subharmonic mode is 9 times that of the unperturbed standing wave.  In
such a study, it is essential to account for the nonlinear interaction
of the unstable mode with the carrier wave to understand its
transition back to a nearly recurrent state. If the wavelength of the
perturbation is an irrational multiple of that of the carrier wave,
this is inherently a large-amplitude spatially quasi-periodic dynamics
problem for which weakly nonlinear theory may be insufficient to
maintain accuracy.
%

For larger-amplitude waves, weakly nonlinear theory is not 
an accurate water wave model. The spectral stability of
large-amplitude Stokes waves to subharmonic perturbations has been
studied by Longuet-Higgins \cite{longuet:78}, McLean
\cite{mclean:82}, MacKay and Saffman \cite{mackay:86}, Deconinck and
Oliveras \cite{oliveras:11}, Deconinck et.~al.~\cite{trichtchenko:16},
and many others.  The eigenvalues of the linearized evolution operator
in a Fourier-Bloch stability analysis give growth rates for
small-amplitude subharmonic perturbations. When the growth rate is
positive, our framework for solving the quasi-periodic initial value
problem \cite{quasi:ivp} provides the groundwork needed to account for
nonlinear dynamics once the unstable mode amplitude grows beyond the
realm of validity of the linearization about the Stokes wave.  When
the eigenvalue is zero, the methods of this paper can be used
to follow new branches of quasi-periodic traveling waves
that bifurcate from the main branch of periodic Stokes waves.

Chen and Saffman \cite{chen80a} found wavelength-doubling and
wavelength-tripling bifurcations of this type from finite-amplitude
waves whereas Wilton \cite{wilton1915lxxii, akers2012wilton,
  trichtchenko:16} considered the special case where the bifurcation
occurs at zero-amplitude. Generalizing Wilton's work to the case in
which the linear dispersion relation supports two irrationally related
wave numbers that travel at the same speed, \black{Bridges and Dias
  \cite{bridges1996spatially} used a spatial Hamiltonian structure to
  construct weakly nonlinear approximations of spatially
  quasi-periodic traveling gravity-capillary waves for two special
  cases: deep water and shallow water.  The existence of such waves in
  the fully nonlinear setting is still an open problem.  In this
  paper, we demonstrate their existence numerically and explore their
  properties.}


Beyond long-time dynamics of unstable subharmonic modes and new
branches of traveling waves, spatially quasi-periodic water waves
arise in other ways. Wave forecasting in oceanography is usually based
on Monte Carlo ensemble-averaged sea states, where the surface
elevation is considered as a random variable satisfying certain
probability distributions and the wave spectrum is continuous.  In
numerical simulation \cite{janssen2003nonlinear}, the discretization
of wavenumber space will lead to spatially quasi-periodic waves.
Another way in which spatial and temporal quasi-periodicity can arise
is by approximating the wave dynamics using an integrable model
equation such as NLS, KdV or the Benjamin-Ono equation. These
equations have hierarchies of exact quasi-periodic solutions that
appear when using the inverse scattering transform to represent
solutions \cite{flaschka,dobro:91}. As another example, Torres and
collaborators \cite{torres2003quasiperiodic,torres2006} have
demonstrated that quasi-periodic pattern formation can emerge in a
parametrically driven Faraday wave tank when the container has a
carefully prepared bottom topography. This work was motivated by the
problem of finding an analog of Bloch theory for quasi-crystals in
materials science \cite{shechtman84,levine1984quasicrystals}.

As a starting point for our work, recall the dispersion
relation for linearized traveling gravity-capillary waves in deep
water:
\begin{equation}\label{dispersion_relation}
  c^2 = gk^{-1}+\tau k.
\end{equation}
Here $c$ is the phase speed, $k$ is the wave number, $g$ is the
acceleration due to gravity and $\tau$ is the coefficient of surface
tension. Notice that $c=\sqrt{(g/k)+\tau k}$ has a positive minimum,
denoted by $c_\text{crit}$. For any fixed phase speed
$c>c_\text{crit}$, there are two distinct positive wave numbers
satisfying the dispersion relation (\ref{dispersion_relation}), denoted
$k_1$ and $k_2$. Any traveling solution of the linearized problem with
this speed can be expressed as a superposition of waves with these
wave numbers.
If $k_1$ and $k_2$ are rationally related, the motion is spatially
periodic and corresponds to the well-known Wilton ripples
\cite{wilton1915lxxii, akers2012wilton, trichtchenko:16}.
However, if $k_1$ and $k_2$ are irrationally related, the motion
will be spatially quasi-periodic.

Recently, Berti and Montalto \cite{berti2016quasi} and Baldi
et.~al.~\cite{baldi2018time} have proved the existence of
small-amplitude temporally quasi-periodic gravity-capillary standing
waves using Nash-Moser theory. Using similar techniques,
  Berti~et.~al.~\cite{berti2020traveling} have proved the existence of
  small-amplitude time quasi-periodic traveling gravity-capillary
  waves with constant vorticity; and Feola and Giuliani
  \cite{feola2020trav} have proved existence of time quasi-periodic
  traveling gravity waves without surface tension or vorticity.
  Quasi-periodic traveling waves have a special meaning in the latter
  two papers that does not imply that they evolve without changing
  shape.  All four papers formulate the problem on a spatially periodic domain,
  and it is shown that solutions of the linearized standing wave or
  traveling wave problems can be combined and perturbed to obtain
temporally quasi-periodic solutions of the nonlinear problem.
Following the same philosophy, we look for spatially quasi-periodic
solutions of the traveling water wave equations that are perturbations
of solutions of the linearized problem. The velocity potential can be
eliminated from the Euler equations when looking for traveling
solutions, so our goal is to study traveling waves with height
functions of the form
\begin{equation}\label{quasi_form}
  \eta(\alpha) = \tilde\eta(k_1\alpha,k_2\alpha), \qquad
  \tilde\eta(\alpha_1,\alpha_2)=
  \sum_{(j_1, j_2)\in\mbb Z^2}\hat{\eta}_{j_1, j_2}e^{i(j_1\alpha_1+j_2\alpha_2)}.
\end{equation}
Here $\tilde\eta$ is real-valued and defined on the torus $\mbb
T^2=\mbb R^2/2\pi\mbb Z^2$, and $\alpha$ parametrizes the free surface
in such a way that the fluid domain is the image of the lower
half-plane $\{w=\alpha+i\beta\,:\,\beta<0\}$ under a conformal map
$z(w)$ whose imaginary part on the upper boundary is
$\im\{z\vert_{\beta=0}\}=\eta$.  The leading term here is
$\eta_\text{lin}(\alpha)=2\re\{\hat\eta_{1,0}e^{ik_1\alpha}+
  \hat\eta_{0,1}e^{ik_2\alpha}\}$, which will be a solution of the
linearized problem.

Unlike \cite{bridges1996spatially}, we use a
conformal mapping formulation \cite{dyachenko1996analytical,
  dyachenko1996nonlinear, 
  choi1999exact, dyachenko2001dynamics, zakharov2002new,
  li2004numerical, hunter2016two, dyachenko:S:2019} 
  of the gravity-capillary water wave
problem.  This makes it possible to compute the normal velocity of the
fluid from the velocity potential on the free surface via a
quasi-periodic variant of the Hilbert transform. As in the periodic
case, the Hilbert transform is a Fourier multiplier operator, but now
acts on functions defined on a higher-dimensional torus. In a
companion paper \cite{quasi:ivp}, we use this idea to develop a
numerical method to compute the time evolution of solutions of the
Euler equations from arbitrary quasi-periodic initial data. The
present paper focuses on traveling waves in this framework.

We formulate the traveling wave computation as a nonlinear
least-squares problem and use the Levenberg-Marquardt method to search
for solutions. This approach builds on the overdetermined shooting
methods developed by Wilkening and collaborators
\cite{ambrose2010computation, ambrose2014dependence,
  rycroft2013computation, wilkening2012overdetermined,
  govindjee2014cyclic} to compute standing waves and other
time-periodic solutions. Specifically, we fix the ratio $k_2/k_1$,
denoted by $k$, and solve simultaneously for the phase speed $c$, the
coefficient of surface tension $\tau$, and the unknown Fourier modes
$\hat\eta_{j_1,j_2}$ in (\ref{quasi_form}) subject to the constraint
that $\hat\eta_{1,0}$ and $\hat\eta_{0,1}$ have prescribed values.  In
Section~\ref{sec:num}, we discuss the merits of these bifurcation
parameters over, say, prescribing $\tau$ and $\hat\eta_{1,0}$ and
solving for $\hat\eta_{0,1}$ along with $c$ and the other unknown
Fourier modes.  While the numerical method is general and can be used
to search for solutions for any irrational~$k$, for brevity we present
results only for $k=1/\sqrt2$ and $k = \sqrt{151}$, which
exhibit clear nonlinear interaction between the two component waves.

Because we focus here on quasi-periodic traveling waves that
  persist to zero-amplitude, the left and right branches of the
  dispersion relation (\ref{dispersion_relation}) can be viewed as the
  wave numbers of gravity waves and capillary waves, respectively
  \cite{djordjevic77}. For the ocean, the ratio between them would be
  many orders of magnitude larger than we consider here, so our
  results pertain to much smaller-scale laboratory experiments rather
  than the ocean. Staying within the quasi-periodic Wilton ripple
  framework that begins at small amplitude with the dispersion
  relation (\ref{dispersion_relation}) would be problematic for the
  ocean as increasing $k$ to $10^7$ does not seem likely to lead to
  interesting nonlinear interactions between gravity and capillary
  waves due to their vast separation of scales, and is anyway
  computationally out of reach for our current algorithm.

A more
  promising idea is to look for spatially quasi-periodic gravity waves
  (with negligible surface tension) that bifurcate from finite
  amplitude periodic traveling waves, similar to the
  wavelength doubling and tripling bifurcations found by Chen and
  Saffman \cite{chen80a}.  In this case, both component waves are
  gravity waves and the bifurcation arises due to a nonlinear
  resonance in the Euler equations. We have computed such a
  quasi-periodic bifurcation from the family of $2\pi$-periodic ``pure
  gravity'' Stokes waves at a wave height of 0.809070794 and a wave
  speed of 1.083977047 when $k=1/\sqrt2$.  Details on these
  preliminary results will be given in future work. We also hope to
  extend our results to the case of finite-depth water waves, search
  for quasi-periodic perturbations of overhanging traveling
  gravity-capillary waves \cite{akers:ambrose14}, and study the
  stability of these waves \cite{oliveras:11, trichtchenko:16,
    torres2003quasiperiodic}.

The paper is organized as follows.
In Section~\ref{sec:prelim}, we define a quasi-periodic Hilbert
transform, derive the equations of motion governing quasi-periodic
traveling water waves, and summarize the main results and notation
introduced in \cite{quasi:ivp} on the more general spatially
quasi-periodic initial value problem. 
In Section~\ref{sec:num}, we design a Fourier pseudo-spectral method
to numerically solve the torus version of the quasi-periodic traveling
wave equations. The discretization leads to an overdetermined
nonlinear least-squares problem that we solve using a variant of the
Levenberg-Marquardt method \cite{nocedal,wilkening2012overdetermined}.
In Section~\ref{sec:rslts}, we present a detailed numerical study of a
two-parameter family of quasi-periodic traveling waves with
$k=1/\sqrt2$ and $g=1$ and validate the accuracy of the
  method. We then search for larger-amplitude
  waves with $k=1/\sqrt2$ and $k=1/\sqrt{151}$ and explore the
  computational limits of our implementation.
%
%
In the conclusion section, we summarize the results and discuss
the effects of floating point arithmetic and whether solutions might
exist for rational values of $k$.
Finally, in Appendix~\ref{sec:dyn:trav}, we study the dynamics of
quasi-periodic traveling waves and show that the waves maintain
a permanent form but generally travel at a non-uniform speed in
conformal space in order to travel at constant speed in physical
  space.


\section{Preliminaries} \label{sec:prelim}


As explained above, the primary goal of this paper is to study
spatially quasi-periodic traveling water waves using a
conformal mapping framework.  In this section, we establish notation;
review the properties of the quasi-periodic Hilbert transform; discuss
quasi-periodic conformal maps and complex velocity potentials; and
propose a synthesis of viewpoints between the Hou, Lowengrub and
Shelley formalism for evolving interfaces \cite{HLS94,HLS97} and the
conformal mapping method developed by Dyachenko
  et.~al.~\cite{dyachenko1996analytical} and subsequent authors \cite{
    dyachenko1996nonlinear, choi1999exact, dyachenko2001dynamics,
    zakharov2002new, dyachenko:S:2019}. We also summarize the
one-dimensional and torus versions of the equations of motion for the
spatially quasi-periodic initial value problem \cite{quasi:ivp};
discuss families of 1d quasi-periodic solutions corresponding to a
single solution of the torus version of the problem; derive the
equations governing traveling waves; and review the linear theory of
quasi-periodic traveling waves.


\subsection{Quasi-periodic functions and the Hilbert transform}
\label{sec:gov:eqs}

A function $u(\alpha)$ is quasi-periodic if there exists a
continuous, periodic function $\tilde u(\bds\alpha)$ defined on the
$d$-dimensional torus $\mbb T^d$ such that
\begin{equation}\label{general_quasi_form}
  u(\alpha) = \tilde u(\bds{k} \alpha), \qquad
  \tilde u(\bds\alpha) = \sum_{\bds{j}\in\mathbb{Z}^d}\hat{u}_{\bds{j}}
  e^{i\la\bds{j},\,\bds{\alpha}\ra}, \qquad
  \alpha\in\mbb R, \;\; \bds\alpha,\bds k \in \mathbb{R}^d.
\end{equation}
We generally assume $\tilde u(\bds\alpha)$ is real analytic, which
means the Fourier modes satisfy the symmetry condition $\hat u_{-\bds
  j}=\overline{\hat u_{\bds j}}$ and decay exponentially as $|\bds
j|\rightarrow\infty$, i.e.~$|\hat u_{\bds j}|\le Me^{-\sigma|\bds j|}$
for some $M,\sigma>0$. Entries of the vector $\bds{k}$ are required to
be linearly independent over $\mathbb{Z}$.  Fixing this vector $\bds
k$, we define two versions of the Hilbert transform, one acting on $u$
(the quasi-periodic version) and the other on $\tilde u$ (the torus
  version):
\begin{equation}\label{eq:H:def}
  H[u](\alpha) = \frac{1}{\pi}\PV
  \int_{-\infty}^\infty\frac{u(\xi)}{\alpha-\xi}\,d\xi, \qquad
  H[\tilde u](\bds\alpha) = \sum\limits_{\bds{j}\in\mathbb{Z}^d}
  (-i)\sgn(\la\bds{j},\,\bds{k}\ra) \hat{u}_{\bds{j}} e^{i
    \la\bds{j},\,\bds{\alpha}\ra}.
\end{equation}
Here $\sgn(q)\in\{1,0,-1\}$ depending on whether $q>0$, $q=0$ or
$q<0$, respectively.  Note that the torus version of $H$ is a Fourier
multiplier on $L^2(\mbb T^d)$ that depends on $\bds k$. It is shown in
\cite{quasi:ivp} that
\begin{equation}\label{eq:H:Htil}
  H[u](\alpha)=H[\tilde u](\bds k\alpha),
\end{equation}
and the most general bounded analytic function $f(w)$ in the lower
half-plane whose real part agrees with $u$ on the real axis has the
form
\begin{equation}\label{eq:f:from:u}
  f(w) = \hat u_{\bds 0} + i\hat v_{\bds 0} + \sum_{\la\bds j,\bds k\ra<0}
  2 \hat u_{\bds j}e^{i\la\bds j,\bds k\ra w}, \qquad
  (w=\alpha+i\beta\,,\; \beta\le0),
\end{equation}
where $\hat v_{\bds 0}$ is an arbitrary constant and the sum is over
all $\bds j\in\mbb Z^d$ satisfying $\la\bds j,\bds k\ra<0$. The
imaginary part of $f$ on the real axis is then given by $v=\hat
v_{\bds 0}-H[u]$. Similarly, given $v$, the most general bounded
analytic function $f(w)$ in the lower half-plane whose imaginary part
agrees with $v$ on the real axis has the form (\ref{eq:f:from:u}) with
$u=\hat u_{\bds 0} + H[v]$, where $\hat u_{\bds 0}$ is an arbitrary
constant. This analytic extension is quasi-periodic on slices
of constant depth, i.e.
\begin{equation}\label{eq:f:tilde}
  f(w) = \tilde f(\bds k\alpha,\beta), \qquad
  (w=\alpha+i\beta\,,\;\beta\le0),
\end{equation}
where $\tilde f(\bds\alpha,\beta) = \hat u_{\bds0} + i\hat v_{\bds0} +
\sum_{\la\bds j,\bds k\ra<0} 2[\hat u_{\bds j}e^{-\la\bds j,\bds
    k\ra\beta}] e^{i\la\bds j,\bds\alpha\ra}$ is periodic in
$\bds\alpha$ for fixed $\beta\le0$.  The torus version of the bounded
analytic extension corresponding to $\tilde u(\bds\alpha+\bds\theta)$ is
simply $\tilde f(\bds\alpha+\bds\theta,\beta)$, which has imaginary part
$\tilde v(\bds\alpha+\bds\theta)$ on the real axis. As a result, the
Hilbert transform commutes with the shift operator,
\begin{equation}\label{eq:H:commute}
    H[\tilde u(\cdot+\bds\theta)](\bds\alpha) =
    H[\tilde u](\bds\alpha+\bds\theta),
\end{equation}
which can also be checked directly from (\ref{eq:H:def}).
We also define quasi-periodic and torus versions of two projection
operators,
\begin{equation}\label{eq:proj}
  P = \operatorname{id} - P_0, \qquad
  P_0 [u] = P_0 [\tilde u] = \hat{u}_{\bds{0}}
  = \frac{1}{(2\pi)^d} \int_{\mathbb{T}^d}
  \tilde u(\bds{\alpha})\, d\alpha_1\dots d\alpha_d,
\end{equation}
where $P_0[u]$ is a constant function on $\mbb R$, $P_0[\tilde u]$ is
a constant function on $\mbb T^d$, and $P[u]$ has zero-mean on $\mbb
R$ in the sense that its torus representation, $P[\tilde u]$, which
satisfies $P[u](\alpha)=P[\tilde u](\bds k\alpha)$, has zero mean on
$\mbb T^d$.


\subsection{A quasi-periodic conformal mapping} \label{sec:conformal_mapping}


For the general initial value problem \cite{quasi:ivp}, we consider a
time-dependent conformal map $z(w,t)$ that maps the lower half-plane
\begin{equation}
  \mbb C^- = \{w=\alpha+i\beta\,:\, \alpha\in\mbb R,\,\beta<0\}
\end{equation}
to the fluid domain $\Omega_f(t)$ that lies below the free surface in
physical space.  At each time $t$, we assume $z(w,t)$ extends
continuously to $\overline{\mbb C^-}$, and in fact is analytic on a
slightly larger half-plane $\mbb C^-_\veps=\{w\,:\,\im w<\veps\}$,
where $\veps>0$ could depend on $t$. The free surface $\Gamma(t)$ is
parametrized by
\begin{equation}
  \zeta(\alpha,t)=\xi(\alpha,t)+i\eta(\alpha,t), \quad (\alpha\in\mbb R\,,\,
  t \text{ fixed}), \qquad\quad \zeta = z\vert_{\beta=0}.
\end{equation}
We assume $\alpha\mapsto\zeta(\alpha,t)$ is injective but do not
assume $\Gamma(t)$ is the graph of a single-valued function of $x$
in the derivation.  An example of a time-dependent
spatially quasi-periodic overturning water wave is computed in
\cite{quasi:ivp}. In future work we will study traveling
  quasi-periodic perturbations of the overhanging periodic traveling
  water waves computed by Akers \textit{et.~al.}
  \cite{akers:ambrose14}.

The conformal map is required to remain a bounded distance
from the identity map in the lower half-plane. Specifically, we
require that
\begin{equation}\label{y_boundary}
  |z(w,t)-w|\le M(t) \quad\qquad (w=\alpha+i\beta\,,\; \beta\le0),
\end{equation}
where $M(t)$ is a uniform bound that could vary in time.  The
Cauchy integral formula implies that $|z_w-1|\le M(t)/|\beta|$, so
at any fixed time,
\begin{equation}\label{eq:zw:lim}
  z_w\rightarrow1 \quad \text{ as } \quad \beta\to-\infty.
\end{equation}
In this paper and its companion \cite{quasi:ivp}, we assume $\eta$ has
two spatial quasi-periods, i.e.~at any time it has the form
(\ref{general_quasi_form}) with $d=2$ and $\bds{k}=[k_1,k_2]^T$.  This
is a major departure from previous work \cite{
  meiron1981applications, 
  dyachenko1996analytical, zakharov2002new, dyachenko2016branch},
where $\eta$ is assumed to be periodic. Through non-dimensionalization,
we may set $k_1=1$ and $k_2 = k$, where $k$ is irrational:
\begin{equation}\label{eq:eta:tilde}
    \eta(\alpha,t) = \tilde\eta(\alpha,k\alpha,t), \qquad
    \tilde\eta(\alpha_1,\alpha_2,t) = \sum_{j_1, j_2\in\mbb Z}
  \hat{\eta}_{j_1, j_2}(t)e^{i(j_1\alpha_1+j_2\alpha_2)}.
\end{equation}
Here $\hat{\eta}_{-j_1, -j_2}(t) = \overline{\hat{\eta}_{j_1,j_2}(t)}$
since $\tilde\eta(\alpha_1,\alpha_2,t)$ is real-valued. Since
$w\mapsto[z(w,t)-w]$ is bounded and analytic on $\mbb C^-$ and its imaginary
part agrees with $\eta$ on the real axis, there is a real number $x_0$
(possibly depending on time) such that
\begin{equation}\label{eq:xi:from:eta}
  \xi(\alpha, t) = \alpha + x_0(t) + \Hilbert[\eta](\alpha, t), \qquad
  \xi_\alpha(\alpha, t) = 1 + \Hilbert[\eta_\alpha](\alpha, t).
\end{equation}
We use a tilde to denote the periodic functions on the torus
that correspond to the quasi-periodic parts of $\xi$, $\zeta$ and
$z$,
\begin{equation}\label{eq:xi:zeta}
  \begin{gathered}
    \xi(\alpha,t) = \alpha + \tilde\xi(\alpha,k\alpha,t), \qquad
    \zeta(\alpha,t) = \alpha + \tilde\zeta(\alpha,k\alpha,t), \\
    z(\alpha+i\beta,t)=\left(\alpha+i\beta \right)+
    \tilde z(\alpha,k\alpha,\beta,t),
    \qquad (\beta\le 0).
  \end{gathered}
\end{equation}
Specifically, $\tilde\xi = x_0(t) + \Hilbert[\tilde\eta]$,
$\tilde\zeta= \tilde\xi + i\tilde\eta$, and
\begin{equation}\label{eq:z:tilde}
  \tilde z(\alpha_1,\alpha_2,\beta,t) =
  x_0(t)+i\hat\eta_{0,0}(t)+\sum_{j_1+j_2k<0}
  \left(2i\hat\eta_{j_1,j_2}(t)e^{-(j_1+j_2k)\beta}\right)
  e^{i(j_1\alpha_1+j_2\alpha_2)}.
\end{equation}
Since the modes $\hat\eta_{j_1,j_2}$ are assumed to decay
exponentially, there is a uniform bound $M(t)$ such that $|\tilde
z(\alpha_1,\alpha_2,\beta,t)|\le M(t)$ for $(\alpha_1,\alpha_2)\in\mbb
T^2$ and $\beta\le 0$.  In \cite{quasi:ivp}, we show that as long as
the free surface $\zeta(\alpha,t)$ does not self-intersect at a given
time $t$, the mapping $w\mapsto z(w,t)$ is an analytic isomorphism of
the lower half-plane onto the fluid region.

\subsection{The complex velocity potential and equations of motion
for the initial value problem}\label{sec:gov:ivp}

Adopting the notation of \cite{quasi:ivp}, let $\Phi^\ph(x,y,t)$
denote the velocity potential in physical space and let
$W^\ph(x+iy,t)= \Phi^\ph(x,y,t)+i\Psi^\ph(x,y,t)$ denote the complex
velocity potential, where $\Psi^\ph$ is the stream function.  Using
the conformal mapping $z(w,t)$, we pull back these functions to the
lower half-plane and define
\begin{equation*}
  W(w,t) = \Phi(\alpha,\beta,t)+i\Psi(\alpha,\beta,t) =
  W^\ph(z(w,t),t), \qquad\quad (w=\alpha+i\beta).
\end{equation*}
We also define
\begin{equation}
  \varphi=\Phi\vert_{\beta=0}, \qquad\quad
  \psi=\Psi\vert_{\beta=0}.
\end{equation}
We assume $\varphi$ is quasi-periodic with the same quasi-periods as
$\eta$,
\begin{equation}\label{eq:phi:tilde}
  \varphi(\alpha, t) = \tilde\varphi(\alpha,k\alpha,t), \qquad
  \tilde\varphi(\alpha_1,\alpha_2,t) =
  \sum_{j_1, j_2\in\mathbb{Z}}
  \hat{\varphi}_{j_1, j_2}(t) e^{i(j_1\alpha_1+j_2\alpha_2)}.
\end{equation}
The fluid velocity $\nabla\Phi^\ph(x,y,t)$ is assumed to decay to zero
as $y\rightarrow-\infty$ (since we work in the lab frame).
Since $dW/dw=(dW^\ph/dz)(dz/dw)$, it follows from (\ref{eq:zw:lim})
that $dW/dw\rightarrow0$ as $\beta\to-\infty$. Thus,
\begin{equation}\label{hilbert_phi}
  \psi_\alpha = -H[\varphi_\alpha], \qquad
  \psi(\alpha, t) = -\Hilbert[\varphi](\alpha, t).
\end{equation}
Here we have assumed $P_0[\varphi] = \hat{\varphi}_{0,0}(t) = 0$ and
$P_0[\psi] = \hat{\psi}_{0,0}(t) = 0$, which is allowed since $\Phi$
and $\Psi$ can be modified by additive constants (or functions of time
  only) without affecting the fluid motion.

Let $U$ and $V$ denote the normal and tangential velocities of the
curve parametrization, respectively; let
$s_\alpha=|\zeta_\alpha|=(\xi_\alpha^2+\eta_\alpha^2)^{1/2}$ denote
the rate at which arclength increases as the curve
$\alpha\mapsto\zeta(\alpha,t)$ is traversed; and let $\theta$ denote
the tangent angle of the curve relative to the horizontal. Then
\begin{equation}
  \zeta_\alpha = s_\alpha e^{i\theta}, \qquad \zeta_t = (V+iU)e^{i\theta}.
\end{equation}
Tracking a fluid particle $x_p(t)+iy_p(t)= \zeta(\alpha_p(t),t)$ on
the free surface, we find that
\begin{equation*}
  \dot x_p=\xi_\alpha\dot\alpha_p+\xi_t=\Phi^\ph_x, \qquad\quad
  \dot y_p=\eta_\alpha\dot\alpha_p+\eta_t=\Phi^\ph_y.
\end{equation*}
Eliminating $\dot\alpha_p$ gives the kinematic condition
\begin{equation}\label{eq:kinematic0}
  U = \zeta_t\cdot\bds{\hat n} = \nabla\Phi^\ph\cdot\bds{\hat n},
\end{equation}
where $\bds{\hat n}=(-\eta_\alpha,\xi_\alpha)/s_\alpha$ is the outward
unit normal to $\Gamma$ and we have identified $\zeta_t$ with the
vector $(\xi_t,\eta_t)$ in $\mbb R^2$.  The general philosophy
proposed by Hou, Lowengrub and Shelley (HLS) \cite{HLS94,HLS97} is
that while (\ref{eq:kinematic0}) constrains the normal velocity $U$ of
the curve to match that of the fluid, the tangential velocity $V$ can
be chosen arbitrarily to improve the mathematical properties of the
representation or the accuracy and stability of the numerical
scheme. Whereas HLS propose choosing $V$ to keep $s_\alpha(t)$
independent of $\alpha$, we interpret the work of Dyachenko
et.~al.~\cite{dyachenko1996analytical} and subsequent authors
\cite{choi1999exact, dyachenko2001dynamics, dyachenko:S:2019,
  zakharov2002new} as choosing $V$ to maintain a conformal
representation. Briefly, since $\Phi^\ph$ and $\Psi^\ph$ satisfy the
Cauchy-Riemann equations, we have
\begin{equation}\label{eq:kinematic}
  -\frac{\psi_\alpha}{s_\alpha} =
  -\frac{\Psi^\ph_x\xi_\alpha + \Psi^\ph_y\eta_\alpha}{s_\alpha} =
  \frac{\Phi^\ph_y\xi_\alpha - \Phi^\ph_x\eta_\alpha}{s_\alpha} =
  \nabla\Phi^\ph\cdot \bds{\hat n} = U.
\end{equation}
Assuming $z_t/z_\alpha$ is bounded and analytic in the lower
half-plane (justified below),
\begin{equation}
  \frac{z_t}{z_\alpha}\bigg\vert_{\beta=0} =
\frac{\zeta_t}{\zeta_\alpha} =
\frac{V+iU}{s_\alpha} \quad
\Rightarrow \quad
\frac{V}{s_\alpha} = H\left(\frac{U}{s_\alpha}\right)+C_1 =
-H\left(\frac{\psi_\alpha}{s_\alpha^2}\right)+C_1,
\end{equation}
where $C_1$ is an arbitrary constant (in space) that we are free to choose.
For any differentiable function $\alpha_0(t)$,
  replacing $C_1(t)$ by $C_1(t)-\alpha_0'(t)$ will cause a reparametrization
  of the solution with $\alpha$ replaced by $\alpha-\alpha_0(t)$; see
  Appendix~\ref{sec:dyn:trav}.  The
tangential and normal velocities can be rotated back to obtain $\xi_t$
and $\eta_t$ via
\begin{equation}\label{eq:xi:t:eta:t}
  \begin{pmatrix}
    \xi_t \\ \eta_t
  \end{pmatrix} =
  \begin{pmatrix}
    \xi_\alpha & -\eta_\alpha \\
    \eta_\alpha & \xi_\alpha
  \end{pmatrix} 
  \begin{pmatrix}
    V/s_\alpha \\ U/s_\alpha
  \end{pmatrix},
\end{equation}
which can be interpreted as the real and imaginary parts of the
complex multiplication $\zeta_t =
(\zeta_\alpha)(\zeta_t/\zeta_\alpha)$. As explained in
\cite{quasi:ivp}, the first equation of (\ref{eq:xi:t:eta:t}) is
automatically satisfied if the second equation holds and $\xi$ is
reconstructed from $\eta$ via (\ref{eq:xi:from:eta}), provided
$x_0(t)$ satisfies
\begin{equation}\label{eq:x0:evol}
  \frac{dx_0}{dt} = P_0\left[
    \xi_\alpha \frac V{s_\alpha} - \eta_\alpha \frac U{s_\alpha} \right].
\end{equation}
The equations of motion for water waves in the conformal framework may
now be written
\begin{equation}\label{general_conformal}
  \begin{gathered}
    \xi_\alpha = 1 + \Hilbert[\eta_\alpha], \qquad
    \psi = -\Hilbert[\varphi], \qquad
    J = \xi_\alpha^2 + \eta_\alpha^2, \qquad
    \chi = \frac{\psi_\alpha}{J}, \\
    \text{
        choose $C_1$ (see below), \qquad
        compute\, $\jd\frac{dx_0}{dt}$\, in
      (\ref{eq:x0:evol}) if necessary,} \\
    \eta_t = -\eta_\alpha \Hilbert[\chi] - \xi_\alpha\chi +
    C_1\eta_\alpha, \qquad
    \kappa = \frac{\xi_\alpha\eta_{\alpha\alpha} -
      \eta_\alpha\xi_{\alpha\alpha}}{J^{3/2}}, \\[-4pt]
    \varphi_t = P\bigg[\frac{\psi_\alpha^2 - \varphi_\alpha^2}{2J} -
    \varphi_\alpha\Hilbert[\chi] + C_1\varphi_\alpha - g\eta +
    \tau\kappa\bigg],
  \end{gathered}
\end{equation}
where the last equation comes from the unsteady Bernoulli equation and
the Laplace-Young condition for the pressure. These equations govern
the evolution of $x_0$, $\eta$ and $\varphi$. The full curve
$\zeta=\xi+i\eta$ and its analytic extension $z$ to the lower
half-plane can be reconstructed from $\eta$ using
(\ref{eq:xi:from:eta}) and (\ref{eq:xi:zeta}). Doing so ensures that
$z$ is injective and that $z_t/z_\alpha$ remains bounded in the lower
half-plane provided that the free surface does not self-intersect and
$J$ remains nonzero on the surface; see \cite{quasi:ivp} for details.

As noted in \cite{quasi:ivp}, equations (\ref{general_conformal}) can
be interpreted as an evolution equation for the functions
$\tilde\zeta(\alpha_1,\alpha_2,t)$ and
$\tilde\varphi(\alpha_1,\alpha_2,t)$ on the torus $\mbb T^2$.  The
$\alpha$-derivatives are replaced by the directional derivatives
$[\pa_{\alpha_1}+k\pa_{\alpha_2}]$ and the quasi-periodic Hilbert
transform is replaced by its torus version, i.e.~$H[\tilde u]$ in
(\ref{eq:H:def}) above. The pseudo-spectral method proposed in
\cite{quasi:ivp} is based on this representation. A convenient choice
of $C_1$ is
\begin{equation}\label{eq:C1:opt2}
  C_1 = \left[\Hilbert\left(\frac{\tilde\psi_\alpha}{\tilde J}\right) -
    \frac{\tilde\eta_\alpha\tilde\psi_\alpha}{
      (1+\tilde\xi_\alpha)\tilde J}
    \right]_{(\alpha_1,\alpha_2)=(0,0)},
\end{equation}
which causes $\tilde\xi(0,0,t)$ to remain constant in time,
alleviating the need to evolve $x_0(t)$ explicitly.  Here $\tilde
J=(1+\tilde\xi_\alpha)^2 + \tilde\eta_\alpha^2$,
and all instances of $\xi_\alpha$ in (\ref{general_conformal})
must be replaced by
\begin{equation}
  \widetilde{\xi_\alpha} = 1 + \tilde\xi_\alpha
\end{equation}
since the secular growth term $\alpha$ is not part of $\tilde\xi$ in
(\ref{eq:xi:zeta}).  Using (\ref{eq:xi:from:eta}) and (\ref{eq:xi:zeta}),
$\tilde\zeta$ is completely determined by $x_0(t)$ and $\tilde\eta$,
so only these have to be evolved --- the formula for $\tilde\xi_t$ in
(\ref{eq:xi:t:eta:t}) is redundant as long as (\ref{eq:x0:evol}) is
satisfied. Other choices of $C_1$ are considered in
  Appendix~\ref{sec:dyn:trav}.

It is shown in \cite{quasi:ivp} that solving the torus version
of (\ref{general_conformal}) yields a three-parameter family of
one-dimensional solutions of the form
\begin{equation}\label{eq:family:full}
  \begin{aligned}
    \zeta(\alpha,t\,;\,\theta_1,\theta_2,\delta) &= \alpha +
    \delta + \tilde\zeta(
      \theta_1+\alpha,\theta_2+k\alpha,t), \\
    \varphi(\alpha,t\,;\,\theta_1,\theta_2) &= \tilde\varphi(
      \theta_1+\alpha,\theta_2+k\alpha,t),
  \end{aligned} \qquad
  \left( \begin{aligned}
      \alpha\in\mbb R, \, t\ge0 \\
      \theta_1,\theta_2,\delta\in\mbb R
    \end{aligned}\right).
\end{equation}
The parameters $(\theta_1,\theta_2,\delta)$ lead to the same solution
as $(0,\theta_2-k\theta_1,0)$ up to a spatial phase shift and
$\alpha$-reparametrization.  Thus, every solution is equivalent
to one of the form
\begin{equation}\label{eq:smaller:family}
  \begin{aligned}
    \zeta(\alpha,t\,;\,0,\theta,0) &=
    \alpha + \tilde\zeta(\alpha,\theta+k\alpha,t), \\
    \varphi(\alpha,t\,;\,0,\theta) &=
    \tilde\varphi(\alpha,\theta+k\alpha,t)
  \end{aligned}
  \qquad
  \alpha\in\mbb R\,,\;
  t\ge 0\,,\;
  \theta\in[0,2\pi).
\end{equation}
Within this smaller family, two values of $\theta$ lead to
equivalent solutions if they differ by $2\pi(n_1k+n_2)$ for some
integers $n_1$ and $n_2$. This equivalence is due to
solutions ``wrapping around'' the torus with a spatial shift,
\begin{equation}\label{eq:wrap:around}
  \zeta(\alpha+2\pi n_1,t\,;\, 0,\theta,0) =
  \zeta(\alpha,t\,;\,0,\theta+2\pi(n_1k+n_2),2\pi n_1), \quad
  \big(\alpha\in[0,2\pi),\;n_1\in\mbb Z\big).
\end{equation}
Here $n_2$ is chosen so that $0\le[\theta+2\pi(n_1k+n_2)]<2\pi$ and we
used periodicity of $\zeta(\alpha,t\,;\,\theta_1,\theta_2,\delta)$
with respect to $\theta_1$ and $\theta_2$.

It is shown in \cite{quasi:ivp} that if all the waves in the family
(\ref{eq:smaller:family}) are single-valued and have no vertical
tangent lines, there is a corresponding family of solutions of the
Euler equations in a standard graph-based formulation
\cite{zakharov1968stability, CraigSulem, johnson97} that are
quasi-periodic in physical space.

\subsection{Quasi-periodic traveling water waves}
\label{sec:gov:trav}

We now specialize to the case of quasi-periodic traveling waves and
derive the equations of motion in a conformal mapping framework.  One
approach (see e.g.~\cite{milewski:10} for the periodic case) is to
write down the equations of motion in a graph-based representation of
the surface variables $\eta^\ph(x,t)$ and
$\varphi^\ph(x,t)=\Phi^\ph(x,\eta(x,t),t)$ and substitute
$\eta^\ph_t = -c\eta^\ph_x$, $\varphi^\ph_t =
-c\varphi^\ph_x$ to solve for the initial condition of a solution of
the form
\begin{equation}\label{eq:trav:graph}
  \eta^\ph(x,t)=\eta^\ph_0(x-ct), \qquad
  \varphi^\ph(x,t)=\varphi^\ph_0(x-ct).
\end{equation}
We present below an alternative derivation of the equations in
\cite{milewski:10} that is more direct and does not assume the wave
profile is single-valued. Other systems of equations have also been
derived to describe traveling water waves, e.g.~by Nekrasov
\cite{nekrasov1921steady,milne:thomson:hydro} and Dyachenko
\emph{et.}~\emph{al.} \cite{dyachenko2016branch}.

Recall the kinematic condition (\ref{eq:kinematic}) that the normal
velocity of the curve is given by $\zeta_t\cdot\bds{\hat n} = U =
-\psi_\alpha/s_\alpha$. Since the wave travels at constant speed $c$
in physical space, there is a reparametrization $\beta(\alpha,t)$ such
that $\zeta(\alpha,t)=\zeta(\beta(\alpha,t),0)+ct$.  Since
$\zeta_\alpha$ is tangent to the curve, the normal velocity is simply
$\zeta_t\cdot\bds{\hat n}=(c,0)\cdot\bds{\hat
  n}=-c\eta_\alpha/s_\alpha$, where we used $\bds{\hat
  n}=(-\eta_\alpha,\xi_\alpha)/s_\alpha$. We conclude that
\begin{equation}\label{phi_psi_eta_xi}
  \psi_\alpha = c\eta_\alpha, \qquad \varphi_\alpha = \Hilbert[\psi_\alpha] =
  c\Hilbert[\eta_\alpha] = c(\xi_\alpha -1).
\end{equation}
This expresses $\psi$ and $\varphi$ (up to additive constants) in
terms of $\eta$ and $\xi=\alpha+x_0+H[\eta]$, leaving only $\eta$ to
be determined.  As in the graph-based approach of
(\ref{eq:trav:graph}) above, it suffices to compute the initial wave
profile at $t=0$ to know the full evolution of the traveling wave
under (\ref{general_conformal}); however, the wave generally travels
at a non-uniform speed in conformal space in order to travel at
constant speed in physical space. This is demonstrated
  in Section~\ref{sec:time:evol} and proved in
  Appendix~\ref{sec:dyn:trav}.

The two-dimensional velocity potential $\Phi^\ph(x,y,t)$ may be
assumed to exist even if the traveling wave possesses overhanging
regions that cause the graph-based representation via $\eta^\ph(x,t)$
and $\varphi^\ph(x,t)$ to break down.
In a moving frame traveling at
constant speed $c$ along with the wave, the free surface will be a
streamline. Let $\breve z=z-ct$ denote position in the moving frame
and note that the complex velocity potential picks up a background
flow term, $\breve W^\ph(\breve z,t)=W^\ph(\breve z+ct,t)-c\breve z$,
and becomes time-independent. We drop $t$ in the notation and define
$\breve W(w)=\breve W^\ph(\breve z(w))$, where $\breve z(w)=z(w,0)$
conformally maps the lower half-plane onto the fluid region of this
stationary problem. We assume $W^\ph(\breve z(\alpha),0)$ is
quasi-periodic with exponentially decaying mode amplitudes, so
\begin{equation*}
  |\breve W(w)+cw|\le|W^\ph(\breve z(w),0)|+c|\breve z(w)-w|
\end{equation*}
is bounded in the lower half-plane. Since the stream function
$\im\{\breve W^\ph(\breve z)\}$ is constant on the free surface, we
may assume $\im\{\breve W(\alpha)\}=0$ for $\alpha\in\mbb R$. The
function $\im\{\breve W(w)+cw\}$ is then bounded and harmonic in the
lower half-plane and satisfies homogeneous Dirichlet boundary
conditions on the real line, so it is zero \cite{axler:harmonic}. Up
to an additive real constant,
\begin{equation}
  \breve W(w)=-cw.
\end{equation}
Thus, $|\breve\nabla\breve\Phi^\ph|^2=|\breve W'(w)/\breve
z'(w)|^2=c^2/J$. Since the free surface is a streamline in
the moving frame, the steady Bernoulli equation
$(1/2)|\breve\nabla\breve\Phi^\ph|^2+g\eta+p/\rho=C$
together with the Laplace-Young condition $p=p_0-\rho\tau\kappa$
on the pressure gives
\begin{equation}\label{eq:trav:conf}
  \begin{gathered}
    \xi_\alpha = 1 + \Hilbert[\eta_\alpha], \quad\quad
    J = \xi_\alpha^2 + \eta_\alpha^2, \\
    \kappa = \frac{\xi_\alpha\eta_{\alpha\alpha} -
      \eta_\alpha\xi_{\alpha\alpha}}{J^{3/2}}, \quad\quad
    \Projection\left[\frac{c^2}{2J} + g\eta -\tau\kappa\right] = 0,
  \end{gathered}
\end{equation}
which is the desired system of equations for $\eta$.

In the quasi-periodic traveling wave problem, we seek a solution of
(\ref{eq:trav:conf}) of the form (\ref{eq:eta:tilde}), except that
$\tilde\eta$ and its Fourier modes will not depend on time.
Like the initial value problem, (\ref{eq:trav:conf}) can be
interpreted as a nonlinear system of equations for
$\tilde\eta(\alpha_1,\alpha_2)$ defined on $\mbb T^2$, where the
$\alpha$-derivatives are replaced by $[\pa_{\alpha_1} + k\pa_{\alpha_2}]$
and the Hilbert transform is replaced by its torus version in
(\ref{eq:H:def}). Without loss of generality, we
assume
\begin{equation}\label{eq:hat:eta00}
  \hat{\eta}_{0,0} = 0.
\end{equation}
We also assume that $\tilde\eta$ is an even, real function of
$(\alpha_1,\alpha_2)$ on $\mbb T^2$. Hence, in our setup, the
Fourier modes of $\tilde\eta$ satisfy
\begin{equation} \label{fourier_mode_symmetry}
  \hat{\eta}_{-j_1, -j_2} = \overline{\hat{\eta}}_{j_1, j_2}, \quad\quad
  \hat{\eta}_{-j_1, -j_2} = \hat{\eta}_{j_1, j_2}, \quad \quad
      (j_1, j_2)\in \mathbb{Z}^2.
\end{equation}
This implies that all the Fourier modes $\hat{\eta}_{j_1, j_2}$ are
real, and causes $\eta(\alpha)=\tilde\eta(\alpha,k\alpha)$ to be
even as well, which is compatible with the symmetry of
(\ref{eq:trav:conf}).  However, as in (\ref{eq:family:full}),
there is a larger family of quasi-periodic traveling solutions
embedded in this solution, namely
\begin{equation}\label{eq:trav:fam:eta}
  \eta(\alpha;\theta) = \tilde\eta(\alpha,\theta+k\alpha).
\end{equation}
As in (\ref{eq:wrap:around}), two values of $\theta$ lead to
equivalent solutions (up to $\alpha$-reparametrization and a
spatial phase shift) if they differ
by $2\pi(n_1k+n_2)$ for some integers $n_1$ and $n_2$. In general,
$\eta(\alpha-\alpha_0;\theta)$ will not be an even function of
$\alpha$ for any choice of $\alpha_0$ unless $\theta=2\pi(n_1k+n_2)$
for some integers $n_1$ and $n_2$. 
In the periodic case, symmetry
breaking traveling water waves have been found by Zufiria
\cite{zufiria:thesis}, though most of the literature is devoted to
periodic traveling waves with even symmetry.

\subsection{Linear theory of quasi-periodic traveling waves.}
\label{linear_theory}
Linearizing (\ref{eq:trav:conf}) around the trivial solution
$\eta(\alpha) = 0$, we obtain,
\begin{equation}\label{trivial_linearization}
  c^2\Hilbert[\delta{\eta}_\alpha] - g\delta{\eta} +
  \tau\delta{\eta}_{\alpha\alpha} = 0,
\end{equation}
where $\delta\eta$ denotes the variation of $\eta$. Substituting
(\ref{quasi_form}) into (\ref{trivial_linearization}), we
obtain a resonance relation for the Fourier modes of $\delta\eta$:
\begin{equation}\label{resonance_fourier}
  \Big(c^2|j_1k_1+j_2k_2| - g - \tau(j_1k_1+j_2k_2)^2\Big)
  \widehat{\delta\eta}_{j_1,j_2} = 0, \quad
  (j_1, j_2) \in \mathbb{Z}^2.
\end{equation}
Note that $j_1k_1+j_2k_2$, which appears in the exponent of the
  Fourier plane wave representation (\ref{quasi_form}), plays the
role of $k$ in the dispersion relation
(\ref{dispersion_relation}). Many families of quasi-periodic
  traveling wave solutions bifurcate from the
  trivial solution even after specifying $k_1$ and $k_2$. Selecting a branch
  amounts to choosing two of the modes
  $\widehat{\delta\eta}_{j_1,j_2}$ to bring in at linear order and
  setting the others to zero in (\ref{resonance_fourier}). In this
  paper, we focus on the case in which $\hat\eta_{1,0}$ and
  $\hat\eta_{0,1}$ enter linearly. This
gives the first-order resonance conditions
\begin{equation}\label{resonance_conformal}
  c^2k_1 - g - \tau k_1^2 = 0, \quad \quad c^2k_2 - g - \tau k_2^2 = 0,
\end{equation}
where $k_1=1$ and $k_2=k$ in our non-dimensionalized setting.
For right-moving waves, we then have $c=\sqrt{g/k_1 + g/k_2}$
  and $\tau=g/(k_1k_2)$.  Any superposition of waves with
dimensionless wave numbers $k_1=1$ and $k_2=k$ traveling with speed
$c=c_\text{lin}$ will solve the linearized problem (\ref{trivial_linearization})
for $\tau=\tau_\text{lin}$. Here we have introduced the notation
$c_\text{lin}=\sqrt{g+g/k}$ and $\tau_\text{lin}=g/k$ to facilitate
the discussion of nonlinear effects below.


\section{Numerical Method}
\label{sec:num}


Equations (\ref{eq:trav:conf}) involve computing derivatives and
Hilbert transforms of quasi-periodic functions that arise in
intermediate computations. Let $f(\alpha)$ denote one of these
functions, and let $\tilde f$ denote the corresponding periodic
function on the torus,
\begin{equation}\label{quasi_f_form}
  f(\alpha) = \tilde f(\alpha,k\alpha), \qquad
  \tilde f(\alpha_1,\alpha_2) = 
  \sum\limits_{j_1, j_2 \in \mathbb{Z}}
  \hat{f}_{j_1, j_2} e^{i(j_1\alpha_1+j_2\alpha_2)}, \qquad
  (\alpha_1, \alpha_2)\in\mathbb{T}^2.
\end{equation}
Each $\tilde f$ that arises is represented by its values on
a uniform $M_1\times M_2$ grid on the torus $\mbb T^2$,
\begin{equation}
  \tilde f_{m_1,m_2} = \tilde f(2\pi m_1/M_1\,,\,2\pi m_2/M_2), \qquad
  (0\le m_1<M_1\,,\,0\le m_2<M_2).
\end{equation}
Products, powers and quotients in (\ref{eq:trav:conf}) are
evaluated pointwise on the grid while derivatives and the
Hilbert transform are computed in Fourier space via
\begin{equation} \label{quasi_deriv_hilbert}
  \begin{gathered}
    \wtil{f_\alpha} (\alpha_1, \alpha_2) = \sum\limits_{j_1, j_2\in \mathbb{Z}}
    i(j_1 + j_2 k)
    \hat{f}_{j_1, j_2} e^{i(j_1\alpha_1+j_2\alpha_2)}, \\
    \wtil{\Hilbert[f]} (\alpha_1, \alpha_2) = \sum\limits_{j_1, j_2\in \mathbb{Z}}
            (-i)\text{sgn}(j_1 + j_2 k) \hat{f}_{j_1, j_2}
            e^{i(j_1\alpha_1+j_2\alpha_2)}.
  \end{gathered}
\end{equation}
We use the `r2c' version of the 2d FFTW library to rapidly
compute the forward and inverse transform given by
\begin{equation}
  \hat f_{j_1,j_2} = \frac1{M_2}\sum_{m_2=0}^{M_2-1}
  \left(\frac1{M_1}\sum_{m_1=0}^{M_1-1}
    \tilde f_{m_1,m_2} e^{-2\pi ij_1m_1/M_1}\right)e^{-2\pi ij_2m_2/M_2}, \quad
  \left(\begin{gathered}
      0\le j_1\le M_1/2 \\
      -M_2/2< j_2\le M_2/2
    \end{gathered}\right).
\end{equation}
The FFTW library actually returns the index
range $0\le j_2<M_2$, but we use $\hat f_{j_1,j_2-M_2}=\hat
f_{j_1,j_2}$ to de-alias the Fourier modes and map the indices
$j_2>M_2/2$ to their correct negative values. The missing entries with
$-M_1/2<j_1<0$ are determined implicitly by
\begin{equation}
  \hat f_{-j_1,-j_2}=\overline{\hat f_{j_1,j_2}}.
\end{equation}
When computing $f_\alpha$ and $\Hilbert[f]$ via
(\ref{quasi_deriv_hilbert}), the Nyquist modes with $j_1=M_1/2$ or
$j_2=M_2/2$ are set to zero, which ensures that the `c2r' transform
reconstructs real-valued functions $\wtil{f_\alpha}$ and
$\wtil{\Hilbert[f]}$ from their Fourier modes.
Further details on this pseudo-spectral representation
are given in \cite{quasi:ivp} in the context of timestepping
the dynamic equations (\ref{general_conformal}).

This pseudo-spectral representation of
  quasi-periodic functions can be generalized to functions with
  quasi-periods larger than two. In this case, one could still use the
  'r2c' and 'c2r' routines in the FFTW library where the function is
  represented by a $d$-dimensional array of Fourier coefficients:
\begin{equation*}
  \tilde{f}_{m_1, m_2, \cdots, m_d} = 
  \sum_{j_1=0}^{M_1-1} 
  \cdots \sum_{j_d=0}^{M_d-1} \hat{f}_{j_1, j_2, \cdots, j_d} 
  e^{2\pi i j_dm_d/M_d} 
  \cdots e^{2\pi i j_1m_1/M_1},
\end{equation*}
where $\tilde{f}_{m_1, m_2, \cdots, m_d}=\tilde f(2\pi
  m_1/M_1,\dots,2\pi m_d/M_d)$ is the value of $\tilde{f}$ evaluated
on a uniform $M_1\times M_2\times\cdots \times M_d$ grid on
$\mathbb{T}^d$.

In \cite{wilkening2012overdetermined}, an overdetermined shooting
algorithm based on the Levenberg-Marquardt method \cite{nocedal} was
proposed for computing standing water waves accurately and
efficiently. Here we adapt this method to compute quasi-periodic
traveling waves instead of standing waves. We first formulate the
problem in a nonlinear least-squares framework.  We consider $\tau$,
$c^2$ (which we denote as $b$) and $\eta$ as unknowns in
(\ref{eq:trav:conf}) and define the residual function
\begin{equation}
  \mathcal{R}[\tau, b, \hat\eta] := P\left[\frac{b}{2\tilde J}
    + g\tilde\eta -\tau \tilde\kappa \right].
\end{equation}
Here, $\hat\eta$ represents the Fourier modes of $\eta$, which are
assumed real via (\ref{fourier_mode_symmetry}); $J$ and $\kappa$
depend on $\eta$ through the auxiliary equations of
(\ref{eq:trav:conf}); and a tilde indicates that
the function is represented on the torus, $\mbb T^2$, as in
(\ref{quasi_f_form}).  We also define the objective function
\begin{equation}
  \mathcal{F} [\tau, b, \hat\eta] := \frac{1}{8\pi^2}\int_{\mathbb{T}^2}
  \mathcal{R}^2[\tau, b, \hat\eta]\,\, d\alpha_1\,d\alpha_2.
\end{equation}
Note that solving (\ref{eq:trav:conf}) is equivalent to finding a zero
of the objective function $\mathcal{F}[\tau, b, \hat\eta]$. The
parameter $k$ in (\ref{quasi_f_form}) is taken to be a fixed,
irrational number when searching for zeros of $\mc{F}$.

In the numerical computation, we truncate the problem to finite
dimensions by varying only the leading Fourier modes
$\hat{\eta}_{j_1, j_2}$ with $|j_1|\le N_1$ and $|j_2|\le N_2$.
We evaluate the residual $\mc R$ (and compute the Fourier
  transforms) on an $M_1\times M_2$ grid, where $M_i\ge2N_i+2$.
The resulting nonlinear least squares problem is
overdetermined because we zero-pad the Fourier modes $\hat{\eta}_{j_1,
  j_2}$ when $|j_1|$ or $|j_2|$ is larger than $N_1$ or
  $N_2$, respectively. Assuming the $\hat\eta_{j_1,j_2}$ are real
(i.e.~that $\eta$ is even) also reduces the number of unknowns
relative to the number of equations, which are enumerated by the
$M_1M_2$ gridpoints without exploiting symmetry.  Guided by the
linear theory of Section \ref{linear_theory}, we fix the two base
Fourier modes $\hat{\eta}_{1,0}$ and $\hat{\eta}_{0,1}$ at nonzero
amplitudes, chosen independently, and minimize $\mc F$ over
  the remaining unknowns via the Levenberg-Marquardt algorithm.

It might seem more natural to prescribe $\tau$ and $\hat{\eta}_{1,0}$
and solve for $\hat{\eta}_{0,1}$ along with $b=c^2$ and the other
unknown Fourier modes of $\eta$. However, since
  $\tau=\tau_\text{lin}=g/k$ is a constant within the linear
  approximation, deviation of $\tau$ from $\tau_\text{lin}$ is a
  higher-order nonlinear effect. This will be confirmed in
  Figure~\ref{tau_c_plot} of Section~\ref{sec:trav_rslts} below. As a
  result, $\tau$ is a poor choice for a continuation parameter near
  the trivial solution in the same way that solving
  $x^2-y^2=(\tau-\tau_\text{lin})$ for $x(\tau,y)$ or $y(\tau,x)$
  leads to problems of existence, uniqueness, and sensitive dependence
  on $\tau$ near $\tau_\text{lin}$.  Beyond the linear regime, one can
  choose any two parameters among $\tau$, $b$ and the Fourier modes
  $\hat\eta_{j_1,j_2}$ to use as continuation parameters.  How well
  they work will depend on the invertibility and condition number of
  the Fr\'echet derivative of $\mc R$ with respect to the remaining
  variables, using the implicit function theorem.
  We also note that the existence of time quasi-periodic water waves
  has only been established rigorously when $\tau$ belongs to a
  Cantor-like set \cite{berti2016quasi, baldi2018time, berti2020traveling}.
  It is possible that small divisors
  \cite{plotnikov01,iooss05,berti2016quasi} and ``near resonances'' in
  the quasi-periodic traveling wave problem will prevent these
  solutions from existing in smooth families.
  
The Levenberg-Marquardt solver requires a linear ordering of the
unknowns.  We enumerate the $\hat\eta_{j_1,j_2}$ so that
lower-frequency modes appear first. As the ``shell index'' $s$ ranges
from 1 to $\max(N_1,N_2)$, we enumerate all the index pairs
$(j_1,j_2)$ with $\max(|j_1|,|j_2|)=s$ before increasing $s$. Within
shell $s$, we proceed clockwise, along straight lines through the
lattice, from $(0,s)$ to $(s,s)$ to $(s,-s)$ to $(1,-s)$.  The other
Fourier modes are known from (\ref{eq:hat:eta00}) and
(\ref{fourier_mode_symmetry}). If $N_1\ne N_2$, we omit
  $(j_1,j_2)$ in the enumeration if $j_1>N_1$ or $j_2>N_2$.  The
total number of modes $\hat\eta_{j_1,j_2}$ indexed in this way is
\begin{equation}
  N_\text{tot} = N_1(2N_2+1)+N_2.
\end{equation}
We replace
$\hat\eta_{1,0}$ by $\tau$ and $\hat\eta_{0,1}$ by $b$ in the list of
unknowns to avoid additional shuffling of the variables when the
prescribed base modes are removed from the list.  Eventually there
are $N_\text{tot}$ parameters to compute, shown here for the
case that $N_2\ge N_1\ge2$:
\begin{equation}\label{eq:p:enum}
  p_1=\tau, \quad p_2=\hat\eta_{1,1}, \quad
  p_3=b, \quad p_4=\hat\eta_{1,-1}, \quad
  p_5=\hat\eta_{0,2}\;\;,\;\; \dots\;\;,\;\;
  p_{N_\text{tot}}=\hat\eta_{1,-N_2}.
\end{equation} 
Re-ordering the arguments of $\mc R$ and $\mc F$, our goal is to
  find $p$ given $\hat\eta_{1,0}$ and $\hat\eta_{0,1}$ such that $\mc
  R[p;\hat\eta_{1,0},\hat\eta_{0,1}]=0$ and $\mc
  F[p;\hat\eta_{1,0},\hat\eta_{0,1}]=0$.  The objective function
$\mathcal{F}$ is evaluated numerically by the trapezoidal rule
approximation over~$\mbb T^2$, which is spectrally accurate:
\begin{equation} \label{numerical_objective}
  \begin{aligned}
    f(p) &= \frac{1}{2} r(p)^Tr(p) \approx
    \mathcal{F}\left[p;\hat\eta_{1,0},\hat\eta_{0,1} \right], \\[5pt]
    r_m(p) &= \frac{\mathcal{R}\left[p;\hat\eta_{1,0},\hat\eta_{0,1}\right]
      (\alpha_{m_1}, \alpha_{m_2})}{\sqrt{M_1M_2}},
  \end{aligned}
    \quad \left(\begin{gathered}
      m = 1+m_1+M_1m_2 \\
      \alpha_{m_i} = 2\pi m_i/M_i
    \end{gathered}\right),
       \quad 0\leq m_i < M_i.
\end{equation}
The parameters $p_j$ are chosen to minimize $f(p)$ using the
Levenberg-Marquardt method \cite{wilkening2012overdetermined,nocedal}. The
method requires a Jacobian matrix $\partial r_m/\partial p_j$,
which we compute by solving the following variational equations:
\begin{equation}\label{eq:variational:eq}
  \begin{gathered}
    \delta{\xi}_\alpha = \Hilbert[\delta{\eta}_\alpha], \qquad\quad
    \delta{J} = 2\left(\xi_\alpha\delta{\xi}_\alpha  +
      \eta_\alpha\delta{\eta}_\alpha \right), \\
    \delta{\kappa} =  -\frac{3}{2}\kappa\frac{\delta J}{J}
    + \frac{1}{J^{3/2}}
    \Big(\delta{\xi}_\alpha\eta_{\alpha\alpha}+
      \xi_\alpha\delta{\eta}_{\alpha\alpha}-
      \delta{\eta}_\alpha\xi_{\alpha\alpha}-
      \eta_{\alpha}\delta{\xi}_{\alpha\alpha}\Big),\\
    \delta{\mathcal{R}} = P\left[\frac{\delta b}{
        2\tilde J}  -
      \frac{1}{2\tilde J^2} b \wtil{\delta{J}} + g\wtil{\delta{\eta}} -
      \delta\tau\tilde\kappa - \tau\wtil{\delta{\kappa}}\right].
  \end{gathered}
\end{equation}
In the last equation, as before, a tilde denotes the torus
  version of a quasi-periodic function. We then have $\der{r_m}{p_j} =
\delta\mc R(\alpha_{m_1},\alpha_{m_2})/\sqrt{M_1M_2}$,
where $m=1+m_1+M_1m_2$ and the
$j$th column of the Jacobian corresponds to setting the perturbation
$\delta\tau$, $\delta b$ or $\delta\hat\eta_{j_1,j_2}$ corresponding
to $p_j$ in (\ref{eq:p:enum}) to 1 and the others to 0.

Like Newton's method, the Levenberg-Marquardt method generates a
sequence of approximate solutions $p^\e0$, $p^\e1$, etc.,~which
terminate when the residual drops below the desired tolerance or fails
to decrease sufficiently. 
If $\max(|\hat\eta_{1,0}|,|\hat\eta_{0,1}|)\le0.01$, we find that
the solution of the linearized problem serves as a good
  initial guess:
\begin{equation} \label{linear_solution}
  \begin{gathered}
    \tilde\eta^\e0(\alpha_1,\alpha_2) =
      \hat\eta_{1,0}(e^{i\alpha_1}+e^{-i\alpha_1})
    + \hat\eta_{0,1}(e^{i\alpha_2}+e^{-i\alpha_2}), \\
    \tau^\e0 = \tau_\text{lin} = g/k, \quad\quad
    b^\e0 = c_\text{lin}^2 = g+g/k.
  \end{gathered}
\end{equation}
We compute larger-amplitude solutions beyond the applicability
  of linear theory using numerical continuation to explore
one-dimensional slices (or paths) through the two-dimensional family
of quasi-periodic traveling waves holding either the ratio
  $\gamma=\hat\eta_{1,0}/\hat\eta_{0,1}$ fixed or one of the modes
  $\hat\eta_{1,0}$, $\hat\eta_{0,1}$ fixed.
  We find that linear extrapolation
from the previous two solutions on a path works well as the starting
guess for the next Levenberg-Marquardt solve.  Details of our
Levenberg-Marquardt implementation, including stopping criteria and a
strategy for delaying the re-computation of the Jacobian, are given in
\cite{wilkening2012overdetermined}.


\section{Numerical Results} \label{sec:rslts}

\subsection{Spatially quasi-periodic traveling waves} \label{sec:trav_rslts}

We now present a detailed numerical study of solutions of
(\ref{eq:trav:conf}) with $k=1/\sqrt2$ and $g=1$ on three
continuation paths corresponding to $\gamma\in\{5,1,0.2\}$, where
$\gamma=\hat\eta_{1,0}/\hat\eta_{0,1}$ is the amplitude ratio of the
prescribed base modes. In each case, we vary the larger of
$\hat\eta_{1,0}$ and $\hat\eta_{0,1}$ from $0.001$ to $0.01$ in
increments of $0.001$.  The initial guess for the first two solutions
on each path are obtained using the linear approximation
(\ref{linear_solution}), which by (\ref{eq:p:enum}) corresponds to
\begin{equation}\label{compute_linear_solution}
    p^\e0_1=\tau^\e0=\sqrt2, \qquad p^\e0_3=b^\e0=1+\sqrt2, \qquad
    p^\e0_j=0, \quad j\not\in\{1,3\}.
\end{equation}
As noted already, the amplitudes $\hat\eta_{1,0}$ and $\hat\eta_{0,1}$
are prescribed --- they are not included among the unknowns.  The
initial guess for the remaining 8 solutions on each continuation path
are obtained from linear extrapolation from the previous two computed
solutions. In all cases, we use $M=60$ for the grid size and
$N=24$ for the Fourier cutoff in each dimension, where we drop the
  subscripts when $M_1=M_2$ and $N_1=N_2$. The nonlinear least-squares
problem involves $M^2=3600$ equations in $N_\text{tot}=1200$ unknowns.

\begin{figure}
\includegraphics[width=\textwidth]{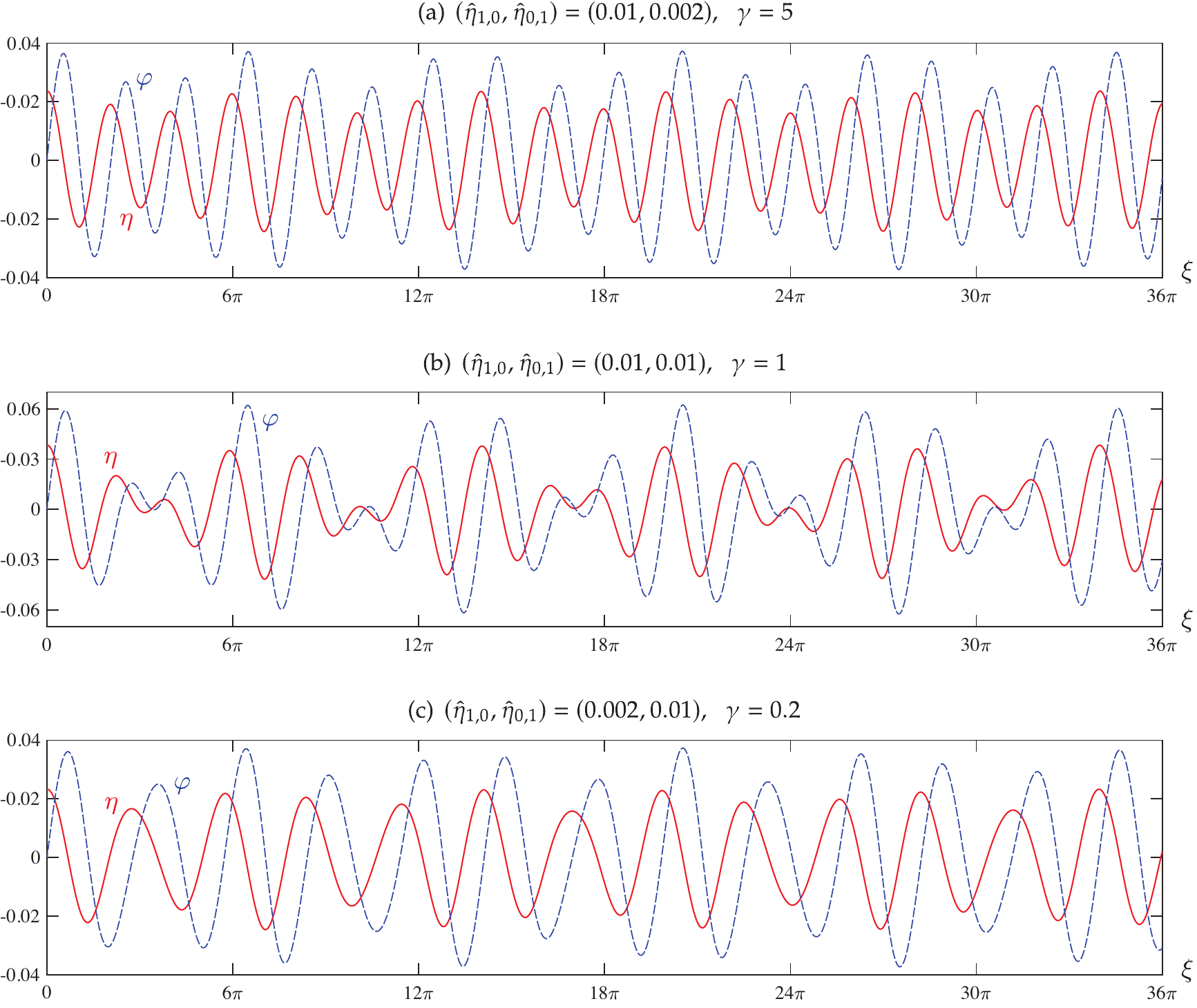}
\caption{\label{initial_plot} Spatially quasi-periodic traveling
  solutions in the lab frame at $t = 0$.  The wave height
  $\eta(\alpha)$ (solid red line) and velocity potential
  $\varphi(\alpha)$ (dashed blue line) are plotted parametrically
  against $\xi(\alpha)$ to show the wave in physical space.}
\end{figure}

Figure~\ref{initial_plot} shows the initial conditions $\eta$ and
$\varphi$ for the last solution on each continuation path (with
  $\max\{\hat\eta_{1,0}\,,\,\hat\eta_{0,1}\}=0.01$).  Panels (a), (b)
and (c) correspond to $\gamma=5,$ $1$, and $0.2$, respectively.  The
solution in all three cases is quasi-periodic, i.e.~$\eta$ and
$\varphi$ never exactly repeat themselves; we plot the solution from
$x=0$ to $x=36\pi$ as a representative snapshot.  For these three
solutions, the objective function $f$ in (\ref{numerical_objective}),
which is a squared error, was
minimized to $6.05\times 10^{-28}$, $9.28\times 10^{-28}$ and
$4.25\times 10^{-28}$, respectively, with
similar or smaller values for lower-amplitude solutions on each path.
For each of the 30 solutions computed on these paths, only
  one Jacobian evaluation and 3--5 $f$ evaluations were needed
  to achieve roundoff-error accuracy.
In our computations, $\eta$ and $\varphi$ are
represented by $\tilde\eta(\alpha_1, \alpha_2)$ and
$\tilde\varphi(\alpha_1, \alpha_2)$, which are defined on the torus
$\mathbb{T}^2$.  In Figure \ref{contour_plot}, we show contour plots
of $\tilde\eta(\alpha_1, \alpha_2)$ and $\tilde\varphi(\alpha_1,
  \alpha_2)$ corresponding to the final solution on each
path. Following the dashed lines through $\mbb T^2$ in
Figure~\ref{contour_plot} leads to the plots in
Figure~\ref{initial_plot}. By construction in (\ref{fourier_mode_symmetry}),
$\tilde\eta(-\bds\alpha)=\tilde\eta(\bds\alpha)$ while
$\tilde\varphi(-\bds\alpha)=-\tilde\varphi(\bds\alpha)$.

\begin{figure}
\includegraphics[width=.8\textwidth]{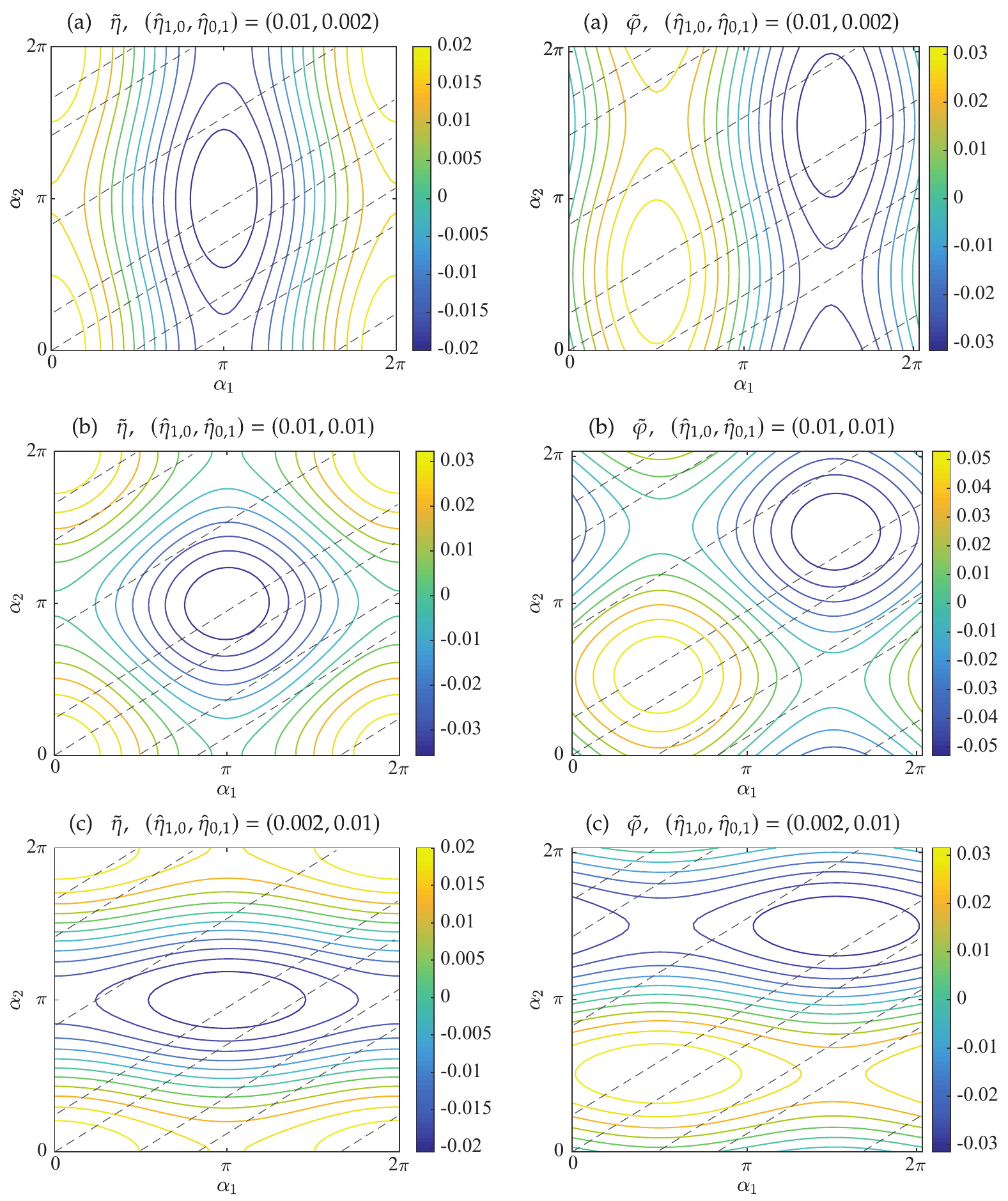}
\caption{\label{contour_plot} Contour plots of $\tilde\eta$ and
  $\tilde\varphi$ on $\mathbb{T}^2$. The dashed lines show $(\alpha,
    k\alpha)$ and its periodic images with $0\le\alpha\le10\pi$ and
  $k=1/\sqrt{2}$.  Evaluating $\tilde\eta$ and $\tilde\varphi$ at
  these points gives $\eta$ and $\varphi$ in (\ref{eq:eta:tilde}) and
  (\ref{eq:phi:tilde}), which were plotted in Figure~\ref{initial_plot}.}
\end{figure}

The amplitude ratio, $\gamma := \hat{\eta}_{1,0} / \hat{\eta}_{0,1}$,
determines the bulk shape of the solution. If $\gamma\gg1$, the component wave
with wave number 1 will be dominant; if $\gamma\ll1$, the component wave with
wave number $k=1/\sqrt2$ will be dominant; and if $\gamma$ is close to
1, both waves together will be dominant over higher-frequency Fourier
modes (at least in the regime we study here). This is demonstrated
with $\gamma=5$, $1$ and $0.2$ in panels (a), (b) and (c) of
Figure~\ref{initial_plot}. Panels (a) and (c) show a clear dominant
mode with visible variations in the amplitude. The oscillations are
faster in panel (a) than in (c) since $1>k\approx0.707$. By contrast,
in panel (b), there is no single dominant wavelength.

This can also be understood from the contour plots of
Figure~\ref{contour_plot}. In case (a), $\gamma\gg1$ and the contour
lines of $\tilde\eta$ and $\tilde\varphi$ are perturbations of
sinusoidal waves depending only on~$\alpha_1$. The unperturbed waves
would have vertical contour lines.  The $\alpha_2$-dependence of the
perturbation causes local extrema to form at the crest and trough. As
a result, the contour lines join to form closed curves that are
elongated vertically since the dominant variation is in the $\alpha_1$
direction. Case (c) is similar, but the contour lines are elongated
horizontally since the dominant variation is in the $\alpha_2$
direction. Following the dashed lines in Figure~\ref{contour_plot}, a
cycle of $\alpha_1$ is completed before a cycle of $\alpha_2$ (since
  $k<1$).  In case (a), a cycle of $\alpha_1$ traverses the dominant
variation of $\tilde\eta$ and $\tilde\varphi$ on the torus, whereas in
case (c), this is true of $\alpha_2$. So the waves in
Figure~\ref{initial_plot} appear to oscillate faster in case (a) than
case (c).  In the intermediate case (b) with $\gamma=1$, the contour
lines of the crests and troughs are nearly circular, but not perfectly
round.  The amplitude of the waves in Figure~\ref{initial_plot} are
largest when the dashed lines in Figure~\ref{contour_plot} pass near
the extrema of $\tilde\eta$ and $\tilde\varphi$, and are smallest when
the dashed lines pass near the zero level sets of $\tilde\eta$ and
$\tilde\varphi$.

Next we examine the behavior of the Fourier modes that make up these
solutions. Figure~\ref{fourier_plot} shows two-dimensional plots of
the Fourier modes $\hat\eta_{j_1,j_2}$ for the 3 cases above, with
$\gamma\in\{5,1,0.2\}$ and $\max\{\hat\eta_{1,0},\hat\eta_{0,1}\}=
0.01$. Only the prescribed modes and the modes that were optimized by
the solver (see (\ref{eq:p:enum})) are plotted, which have indices in
the range $0\le j_1\le N$ and $-N\le j_2\le N$, excluding $j_2\le0$
when $j_1=0$. The other modes are determined by the symmetry of
(\ref{fourier_mode_symmetry}) and by zero-padding
$\hat\eta_{j_1,j_2}=0$ if $N<j_1\le M/2$ or $N<|j_2|\le M/2$.  We used
$N=24$ and $M=60$ in all 3 calculations. One can see that the fixed
Fourier modes $\hat{\eta}_{1,0}$ and $\hat{\eta}_{0,1}$ are the two
highest-amplitude modes in all three cases. In this sense, our
solutions of the nonlinear problem (\ref{eq:trav:conf}) are
small-amplitude perturbations of the solutions
  (\ref{linear_solution}) of the linearized problem. However,
in the plots of Figure~\ref{fourier_plot}, there are many
active Fourier modes other than the four modes $e^{\pm
    i\alpha_1}$, $e^{\pm i \alpha_2}$ from linear theory. In this sense,
  these solutions have left the linear regime.
Carrying out a weakly nonlinear Stokes expansion to high enough order
to accurately predict all these modes would be difficult due to
  the two-dimensional array of unknown Fourier modes, which would
  complicate the analysis of the periodic Wilton ripple problem
  \cite{vandenBroeck:book, trichtchenko:16, akers2020wilton}.
Steeper waves that are well outside of the linear regime
  will be computed in Section~\ref{sec:large:grav:cap}.

\begin{figure}
\includegraphics[width=\textwidth]{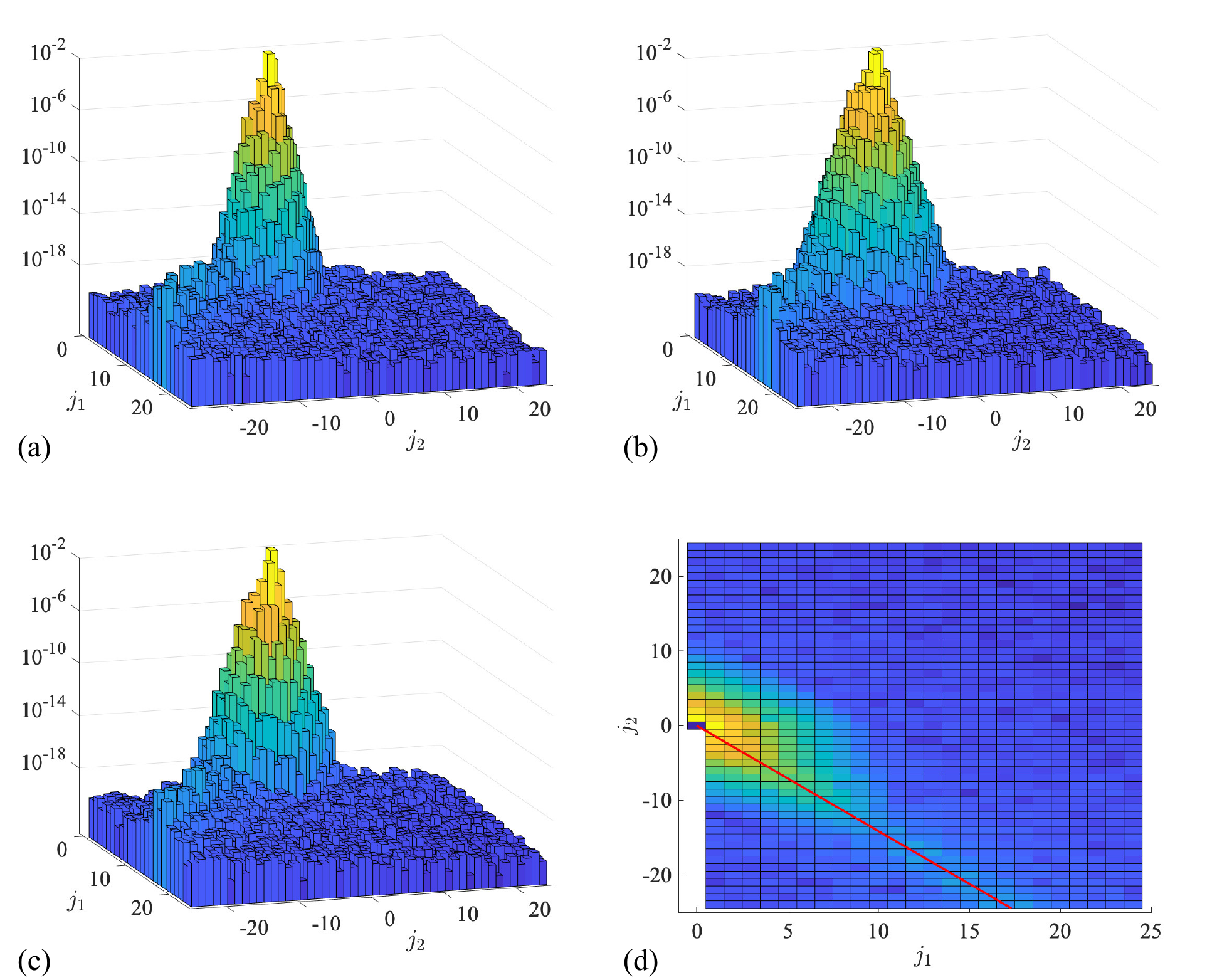}
\caption{\label{fourier_plot} Two-dimensional Fourier modes of
  $\tilde\eta$ for the $k=1/\sqrt2$ solutions plotted in
  Figures~\ref{initial_plot} and~\ref{contour_plot}.  (a) $\gamma=5$.
  (b,d) $\gamma=1$. (c) $\gamma=0.2$. In all three cases, the modes
  decay visibly slower along the line $j_1+j_2k=0$, indicating the
  presence of resonant mode interactions.}
\end{figure}

In panels (a), (b) and (c) of Figure~\ref{fourier_plot}, the modes
appear to decay more slowly in one direction than in other
directions. This is seen more clearly when viewed from above, as shown
in panel (d) for the case of $\gamma=1$. (The other two cases are
  similar). The direction along which the modes decay less rapidly
appears to coincide with the line $\{(j_1,j_2)\;:\;j_1+j_2k=0\}$,
which is plotted in red. A partial explanation is that when $j_1+j_2k$
is close to zero, the corresponding modes $e^{i(j_1+j_2k)\alpha}$ in
the expansion of $\eta(\alpha)$ in (\ref{eq:eta:tilde}) have very long
wavelengths. Slowly varying perturbations lead to small changes in the
residual of the water wave equations, so these modes are not strongly
controlled by the governing equations (\ref{eq:trav:conf}). We believe
this would lead to a small divisor problem that would complicate a
rigorous proof of existence of quasi-periodic traveling water waves.
Similar small divisor problems arise in proving the
  existence of standing water waves \cite{plotnikov01,iooss05}, 3D
  traveling gravity waves \cite{rard2009small}, and 2D time
  quasi-periodic gravity-capillary waves \cite{berti2016quasi,
    baldi2018time, berti2020traveling}, where small divisors
  are tackled using a Nash-Moser iterative scheme.

Next we show that $\tau$ and $c$ depend nonlinearly on the amplitude
of the Fourier modes $\hat{\eta}_{1,0}$ and $\hat{\eta}_{0,1}$.
Panels (a) and (b) of Figure~\ref{tau_c_plot} show plots of $\tau$ and
$c$ versus $\hat\eta_\text{max}:=\max(\hat\eta_{1,0},\hat\eta_{0,1})$
for 9 values of $\gamma=\hat\eta_{1,0}/\hat\eta_{0,1}$, namely
$\gamma=0.1, 0.2, 0.5, 0.8, 1, 1.25, 2, 5, 10$. On each curve,
$\hat\eta_\text{max}$ varies from 0 to $0.01$ in increments of
$0.001$. At small amplitude, linear theory predicts $\tau=g/k=1.41421$
and $c=\sqrt{g(1+1/k)}=1.55377$.  This is represented by the black
marker at $\hat\eta_\text{max}=0$ in each plot. For each value
  of $\gamma$, the curves $\tau$ and $c$ are seen to have zero slope
  at $\hat\eta_\text{max}=0$, and can be concave up or concave down
  depending on $\gamma$. This can be understood from the contour plots
  of panels (e) and (f). Both $\tau$ and $c$ appear to be even
  functions of $\hat\eta_{1,0}$ and $\hat\eta_{0,1}$ when the other is
  held constant. Both plots have a saddle point at the origin, are
  concave down in the $\hat\eta_{1,0}$ direction holding
  $\hat\eta_{0,1}$ fixed, and are concave up in the $\hat\eta_{0,1}$
  direction holding $\hat\eta_{1,0}$ fixed. The solid lines in the
  first quadrant of these plots are the slices corresponding to the
  values of $\gamma$ plotted in panels (a) and (b). The concavity of
  the 1d plots depends on how these lines intersect the saddle in the
  2d plots.

\begin{figure}
\includegraphics[width=\textwidth]{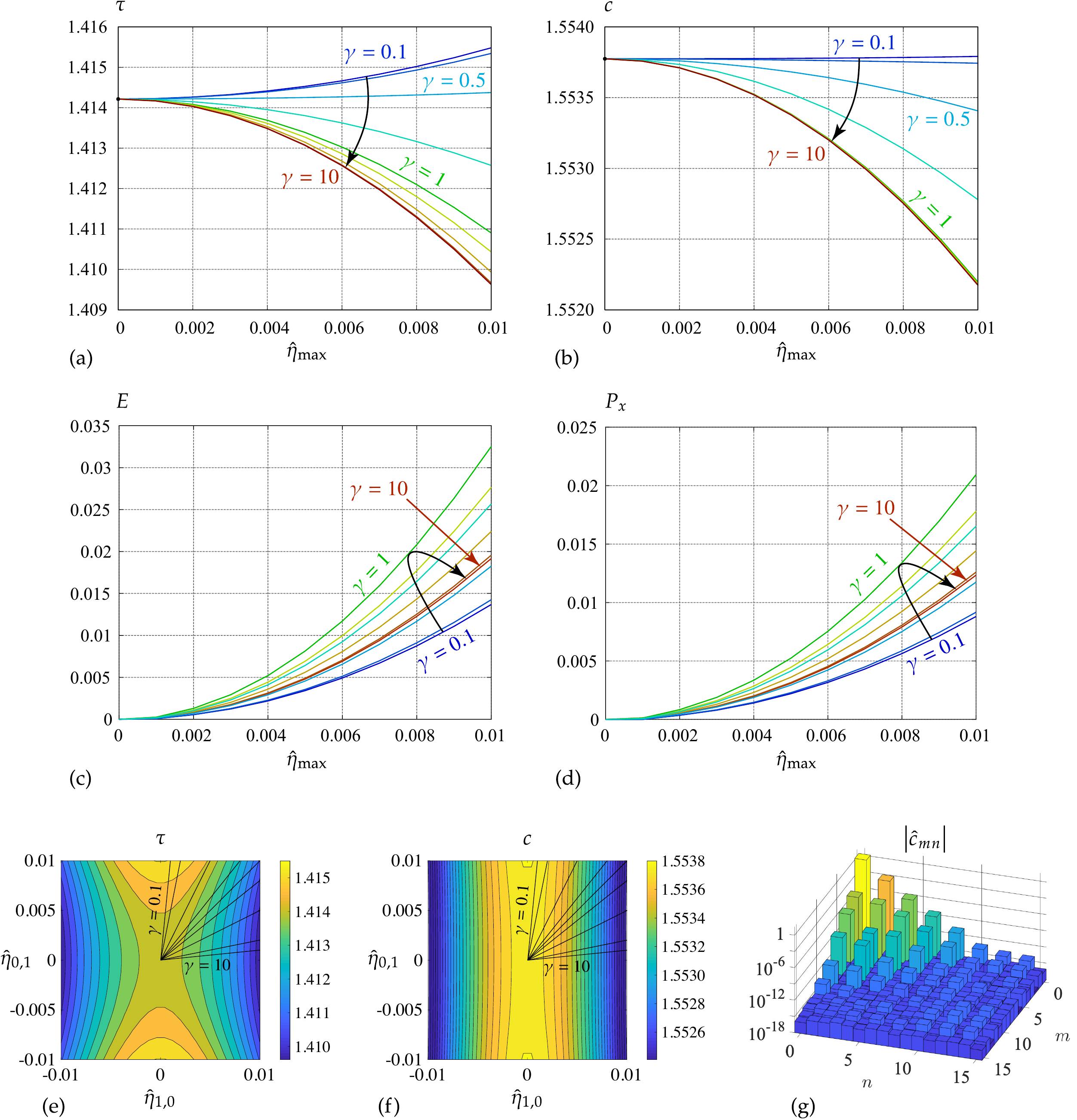}
\caption{\label{tau_c_plot} Surface tension, wave speed, energy
    and momentum of small-amplitude quasi-periodic water waves with
    $k=1/\sqrt2$.  (a,b,c,d) Plots of $\tau$, $c$, $E$ and
  $P_x$ versus
  $\hat{\eta}_{\max}=\max\{\hat\eta_{1,0},\hat\eta_{0,1}\}$
  holding $\gamma=\hat\eta_{1,0}/\hat\eta_{0,1}$ fixed. The
  black arrow in each plot shows how the curves
  change as $\gamma$ increases from $0.1$ to $10$.
  (e,f)
  Contour plots of $\tau$ and $c$ and the rays of constant $\gamma$
  corresponding to (a,b). (g) Mode amplitudes of a 2d Chebyshev
  expansion of $c(\hat\eta_{1,0},\hat\eta_{0,1})$ over the rectangle
  $-0.01\le \hat\eta_{1,0},\hat\eta_{0,1}\le 0.01$.}
\end{figure}

The contour plots of panels (e) and (f) of Figure~\ref{tau_c_plot}
were made by solving (\ref{eq:trav:conf}) with
$(\hat\eta_{1,0},\hat\eta_{0,1})$ ranging over a uniform $26\times26$
grid on the square $[-0.01,0.01]\times[-0.01,0.01]$. 
Using an even number of gridpoints avoids the degenerate case
where $\hat\eta_{1,0}$
or $\hat\eta_{0,1}$ is zero. At those values, the two-dimensional
family of quasi-periodic solutions meets a sheet of periodic solutions
where $\tau$ or $c$ becomes a free parameter. 
Alternative techniques
would be needed in these degenerate cases to determine the value of
$\tau$ or $c$ from which a periodic traveling wave in the nonlinear
regime bifurcates to a quasi-periodic wave. 
In panel (g), we plot the
magnitude of the Chebyshev coefficients in the expansion
\begin{equation}\label{eq:c:cheb:expand}
  c(\hat\eta_{1,0},\hat\eta_{0,1}) = \sum_{m=0}^{15}\sum_{n=0}^{15}
  \hat c_{mn}T_m(100\hat\eta_{1,0})T_n(100\hat\eta_{0,1}), \qquad
  -0.01\le \hat\eta_{1,0},\hat\eta_{0,1}\le 0.01.
\end{equation}
This was done by evaluating $c$ on a cartesian product of two 16-point
Chebyshev-Lobatto grids over $[-0.01,0.01]$ and using the
one-dimensional Fast Fourier Transform in each direction to compute
the Chebyshev modes.  We see that the modes decay to machine precision
by the time $m+n\ge10$ or so, and only even modes $m$ and $n$ are
active.  The plot for $|\hat\tau_{mn}|$ is very similar, so we omit
it.  These plots confirm the visual observation from the contour plots
that $\tau$ and $c$ are even functions of $\hat\eta_{1,0}$ and
$\hat\eta_{0,1}$ when the other is held constant. These properties
of $\tau$ and $c$ make them unsuitable as continuation parameters
near the trivial solution, as discussed in Section~\ref{sec:num}.

In panels (c) and (d) of Figure~\ref{tau_c_plot}, we show the
  energy $E$ and momentum $P_x$ of waves in the above two-parameter
  family of quasi-periodic solutions,
\begin{equation}\label{eq:E:Px:def}
\begin{aligned}
  E = &\int_{\mathbb{T}^2} \frac{1}{2} \tilde{\psi}(\partial_{\alpha_1}
    + k\partial_{\alpha_2})\tilde{\varphi}
  + \frac{1}{2} g\tilde{\eta}^2\big(1+ (\partial_{\alpha_1}
      + k\partial_{\alpha_2})\tilde{\xi}\big) \\
  & \quad + \tau\left(\sqrt{\big(1+ (\partial_{\alpha_1}
          + k\partial_{\alpha_2})\tilde{\xi}\big)^2 + 
      \big((\partial_{\alpha_1} + k\partial_{\alpha_2})\tilde{\eta}\big)^2}
    - 1\right)d\alpha_1d\alpha_2, \\
  P_x = &-\int_{\mathbb{T}^2} \tilde{\varphi}(\partial_{\alpha_1}
    + k\partial_{\alpha_2})\tilde{\eta}\,
  d\alpha_1 d\alpha_2.
\end{aligned}
\end{equation}
These formulas are derived in \cite{zakharov2002new,dyachenko2019} in
the conformal mapping framework for a water wave of infinite depth.
The only modification needed for spatially quasi-periodic waves with
$d$ quasi-periods is that integrals over $\mbb R$ or $\mbb T$ are
replaced by integrals over~$\mbb T^d$. In \cite{quasi:ivp}, it is
confirmed that $E$ and $P_x$ in (\ref{eq:E:Px:def}) are conserved quantities
under the evolution equations (\ref{general_conformal}).
We see in Figure~\ref{tau_c_plot} that the energy and momentum of the
quasi-periodic waves are positively correlated. In particular, the
quasi-periodic wave family with $\gamma = 1$ possesses the largest
energy and momentum when $\hat{\eta}_{\max}$ is fixed, even though it
does not have the highest wave speed. Energy and momentum can both be
regarded as measures of the amplitude of the wave. Unlike the wave
speed, they are both zero at the flat rest state. We note that
$\gamma=1$ corresponds to maximizing both $|\eta_{1,0}|$ and
$|\eta_{0,1}|$ to have the value $\hat{\eta}_{\max}$, and also leads to
the largest amplitude oscillations in Figure~\ref{initial_plot}. The
Hamiltonian structure of the equations of motion could be useful e.g.~in
generalizing the time quasi-periodic results of
Berti~et.~al.~\cite{berti2020traveling}
to the spatially quasi-periodic setting.

\subsection{Time evolution of spatially quasi-periodic traveling waves}
\label{sec:time:evol}

\begin{figure}
\includegraphics[width=.9\textwidth]{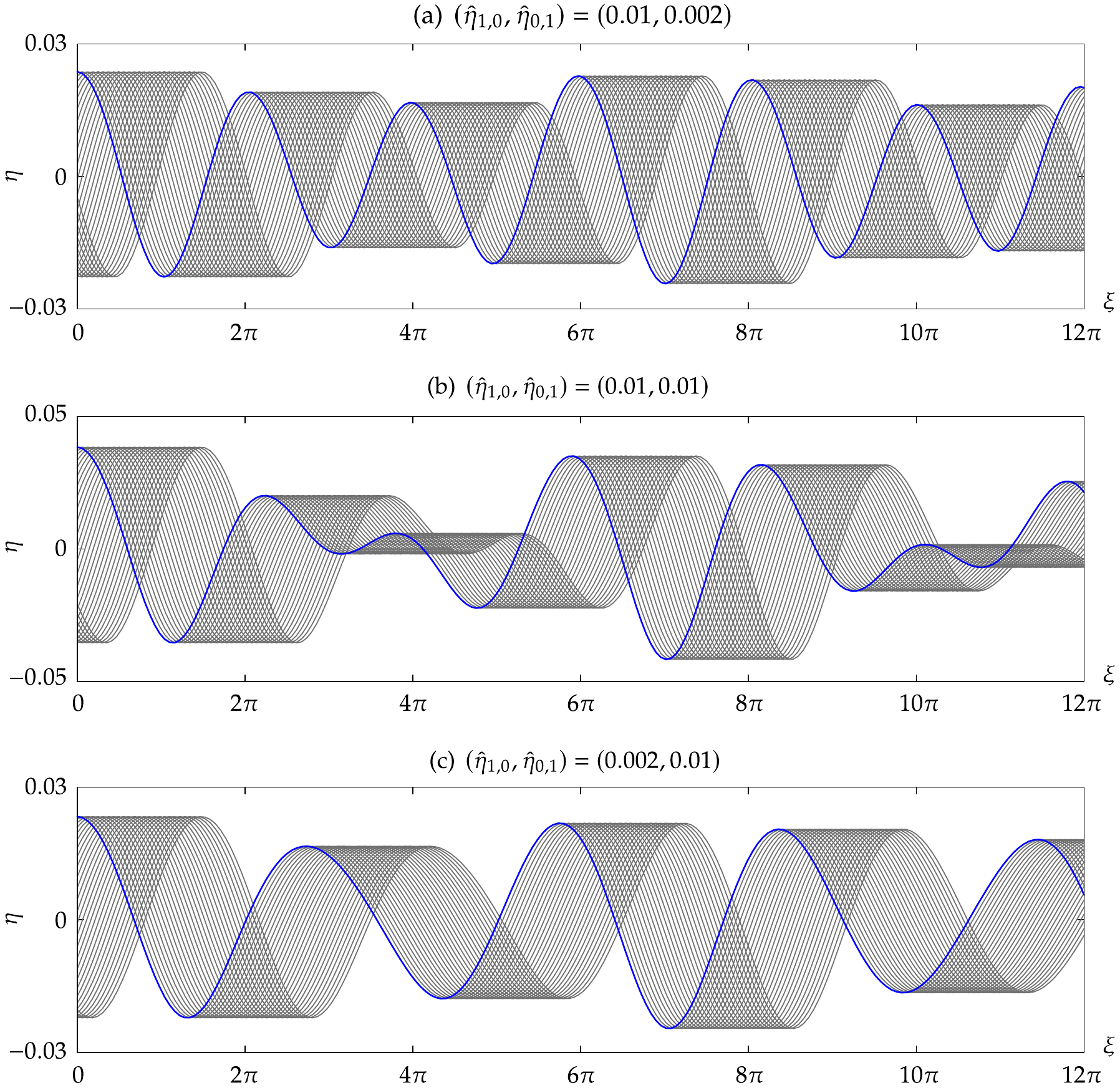}
  \caption{\label{traveling_plot} Time evolution of the traveling wave
    profiles, $\zeta(\alpha,t)$, from $t=0$ to $t=3$ in the lab
    frame. The thick blue lines correspond to the initial conditions.}
\end{figure}

In this section, we confirm that the quasi-periodic solutions
we obtain by minimizing the objective function
(\ref{numerical_objective}) are indeed traveling waves under the
evolution equations (\ref{general_conformal}). This allows us to
  measure the accuracy of our independent codes for solving these two
  problems by comparing the numerical results. An interesting feature
  of the conformal mapping formulation arises in this comparison,
  namely that for most choices of $C_1$ in (\ref{general_conformal}),
  traveling waves move at a non-uniform speed through conformal space
  in order to travel at constant speed in physical space. This is
  discussed in this section and proved in Appendix~\ref{sec:dyn:trav}.

In Figure~\ref{traveling_plot}, we plot the time evolution of
$\zeta(\alpha,t)$ in the lab frame from $t=0$ to $t=3$.
The initial conditions, plotted with thick blue lines, are those of
the traveling waves computed in Figures~\ref{initial_plot}
and~\ref{contour_plot} above by minimizing the objective
  function (\ref{numerical_objective}). The grey curves give
snapshots of the solution at uniformly sampled times with $\Delta
t=0.1$.  They were computed using the 5th order explicit
  Runge-Kutta method described in \cite{quasi:ivp} with a stepsize
  of 1/300, so there are 30 Runge-Kutta steps
between snapshots in the figure. The solutions are plotted over the
representative interval $0\le x\le 12\pi$, though they extend in both
directions to $\pm\infty$ without exactly repeating. The
  initial condition and time evolution were computed on the torus
  and then sampled along the $(1,k)$ direction to extract the data
  for these 1D plots.

For quantitative comparison, let $\tilde\eta_0(\bds\alpha)$
  denote the initial condition on the torus, which is computed
  numerically by minimizing (\ref{numerical_objective}).  We then
compute $\tilde\xi_0=\Hilbert[\tilde\eta_0]$ and
$\tilde\varphi_0=c\tilde\xi_0$, which are odd functions of
$\bds\alpha=(\alpha_1,\alpha_2)\in\mbb T^2$ since $\tilde\eta$ is even. From
Corollary~\ref{cor:trav:conf} of Appendix~\ref{sec:dyn:trav}, we
define the ``exact solution'' of the time evolution of the traveling
wave under (\ref{general_conformal}) and (\ref{eq:C1:opt2}) with these
initial conditions as
\begin{equation}\label{eq:exact:a0}
  \begin{aligned}
    \tilde\eta_\text{exact}(\bds\alpha,t) &=
    \tilde\eta_0\big(\bds\alpha - \bds k \alpha_0(t)\big), \\
    \tilde\varphi_\text{exact}(\bds\alpha,t) &=
    \tilde\varphi_0\big(\bds\alpha - \bds k\alpha_0(t)\big),
  \end{aligned}
\end{equation}
where $\bds k=(1,k)$, $\alpha_0(t)=ct -
\mc A_0(-\bds kct)$ and $\mc A_0(x_1,x_2)$ is a
periodic function on $\mbb T^2$ defined implicitly by
(\ref{eq:mcA:trav}) below. 
We see in (\ref{eq:exact:a0}) that the
waves do not change shape as they move through the torus along the
characteristic direction $\bds k$, but the traveling speed
$\alpha_0'(t)$ in conformal space varies in time in order to maintain
$\tilde\xi(0,0,t)=0$ via (\ref{eq:C1:opt2}). By
Corollary~\ref{cor:trav:conf}, the exact reconstruction of
$\tilde\xi_\text{exact}$ from $\tilde\eta_\text{exact}$ is
\begin{equation}\label{eq:exact:xi}
  \tilde\xi_\text{exact}(\bds\alpha,t) =
    \tilde\xi_0\big(\bds\alpha - \bds k\alpha_0(t)\big) + \delta_0(t),
\end{equation}
where $\delta_0(t) = ct-\alpha_0(t) = \mc A_0(-\bds kct)$ measures the
deviation in position from traveling at the constant speed $ct$ in
conformal space. The defining property (\ref{eq:mcA:trav}) of
$\mc A_0(x_1,x_2)$ ensures that $ \tilde\xi_\text{exact}(0,0,t)=0$.

The significance of $\mc A_0$ is that the inverse of the mapping
  $\bds x = \bds \alpha + \bds k\tilde\xi_0(\bds\alpha)$ on $\mbb
  T^2$, assuming it is single-valued, is
  \begin{equation}\label{eq:mcA:meaning}
  \bds \alpha = \bds x + \bds k\mc A_0(\bds x).
  \end{equation}
  As shown in \cite{quasi:ivp}, this allows us to
  express quasi-periodic solutions of the initial value problem in
  conformal space as quasi-periodic functions in physical space. In
  the traveling case considered here, the exact solutions on the torus
  in physical space are $\tilde\eta_0^\text{phys}(\bds x-\bds kct)$
  and $\tilde\varphi_0^\text{phys}(\bds x-\bds kct)$, where
  e.g.~$\tilde\eta_0^\text{phys}(\bds x) = \tilde\eta_0(\bds x +\bds k
    \mc A_0(\bds x))$. We know this already on physical grounds, but
  it also follows from (\ref{eq:exact:a0}) and (\ref{eq:exact:xi})
  using
  \begin{equation*}
    \tilde\eta^\text{phys}_\text{exact}(\bds x,t) =
    \tilde\eta_\text{exact}(\bds x + \bds k\mc A(\bds x,t)\,,\,t),
    \qquad
    \tilde\varphi^\ph_\text{exact}(\bds x,t) =
    \tilde\varphi_\text{exact}(\bds x + \bds k\mc A(\bds x,t)\,,\,t),
  \end{equation*}
  where $\mc A(\bds x,t) = \mc A_0(\bds x-\bds kct) - \mc A_0(-\bds kct)$
satisfies the time-dependent analog of (\ref{eq:mcA:trav}).

\begin{figure}
\includegraphics[width=\textwidth]{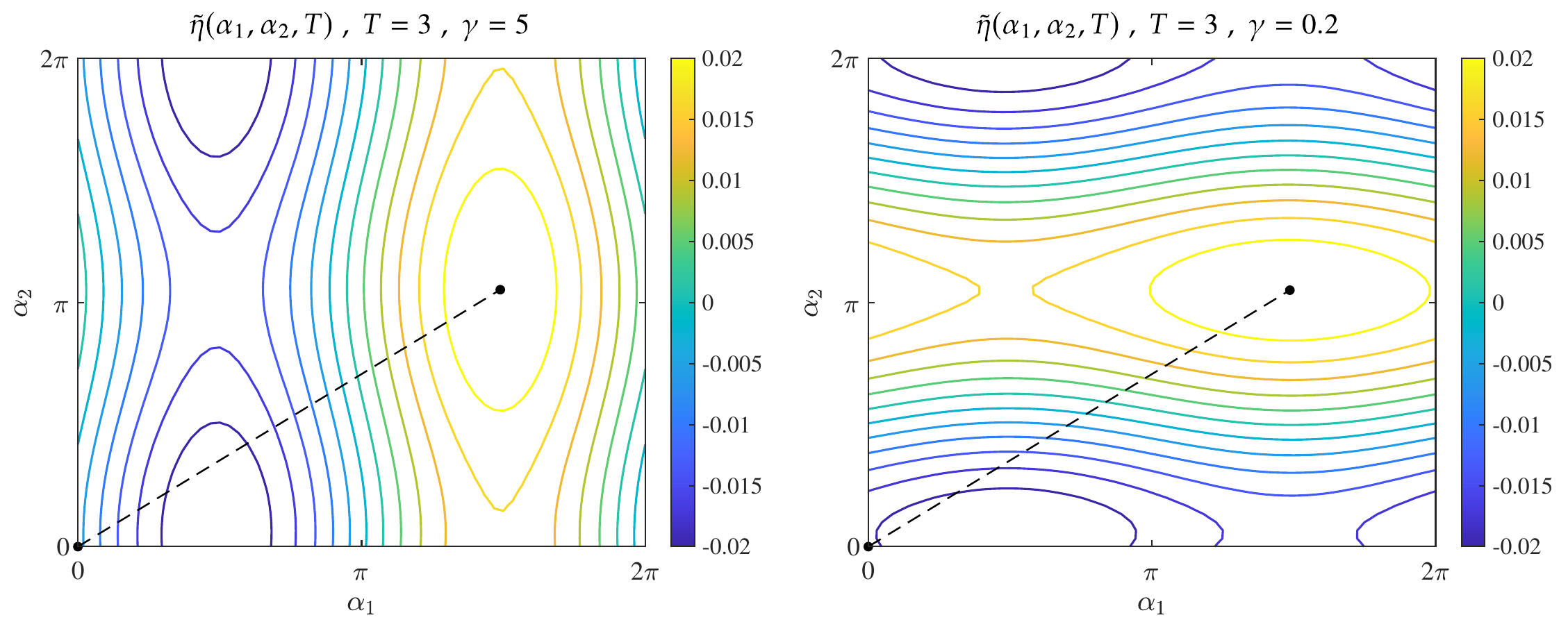}
\caption{\label{stepper_plot} Contour plots of the numerical solution
  $\tilde\eta(\alpha_1,\alpha_2,T)$ on the torus corresponding to the
  quasi-periodic solutions $\eta(\alpha,t)$ of panels (a) and (c) of
  Figure~\ref{traveling_plot} at the final time shown, $t=T=3$. The
  dashed lines show the trajectory of the wave crest from $t=0$ to
  $t=T$.}
\end{figure}

Figure~\ref{stepper_plot} shows contour plots of the torus version of
the $\gamma=5$ and $\gamma=0.2$ solutions shown in panels (a) and (c)
of Figure~\ref{traveling_plot} at the final time computed, $T=3$.  A
similar plot of the $\gamma=1$ solution is given in \cite{quasi:ivp}.
The dashed lines show the trajectory from $t=0$ to $t=T$ of the wave
crest that begins at $(0,0)$ and continues along the path
$\alpha_1=\alpha_0(t),$ $\alpha_2=k\alpha_0(t)$ through the torus in
(\ref{eq:exact:a0}). The following table gives the phase speed, $c$,
surface tension,~$\tau$, translational shift in conformal space at the
final time computed, $\alpha_0(T)$, and deviation from steady motion
in conformal space, $\delta_0(T)$, for these three finite-amplitude
solutions (recall that $\max\{\hat\eta_{1,0},\hat\eta_{0,1}\}=0.01$
  and $\hat\eta_{1,0}/\hat\eta_{0,1}=\gamma$) as well as for the
zero-amplitude limit:
\begin{equation*}
  \begin{array}{c|c|c|c|c}
    & \gamma=5 & \gamma=1 & \gamma = 0.2 & \text{linear theory} \\ \hline
    c & \phm 1.552\,175 & \phm 1.552\,197 & \phm 1.553\,743 &
    c_\text{lin} = 1.553\, 774 \\
    \tau & \phm 1.409\,665 & \phm 1.410\,902 & \phm 1.415\,342 &
    \tau_\text{lin} = 1.414\,214 \\
    \alpha_0(T) & \phm 4.677\,416 & \phm 4.681\,174 & \phm 4.668\,757 &
    c_\text{lin}T = 4.661\,322 \\
    \delta_0(T) & -0.020\,890 & -0.024\,583 & -0.007\,527 & 0
  \end{array} \qquad\quad
  \begin{array}{c} \\ \\ \\ (T=3) \end{array}
\end{equation*}
In Figure~\ref{delta_plot}, we plot $\delta_0(t)$ for $0\le t\le T$
(solid lines) along with $(c-c_\text{lin})t$ (dashed and dotted lines)
for the three finite-amplitude solutions in this table. Writing
$\alpha_0(t) = c_\text{lin}t + [(c-c_\text{lin})t - \delta_0(t)]$, we
see that the deviation of $\alpha_0(t)$ from linear theory over this
time interval is due mostly to fluctuations in $\delta_0(t)$ rather
than the steady drift $(c-c_\text{lin})t$ due to the change in phase
speed $c$ of the finite-amplitude wave.

\begin{figure}
\includegraphics[width=.75\textwidth]{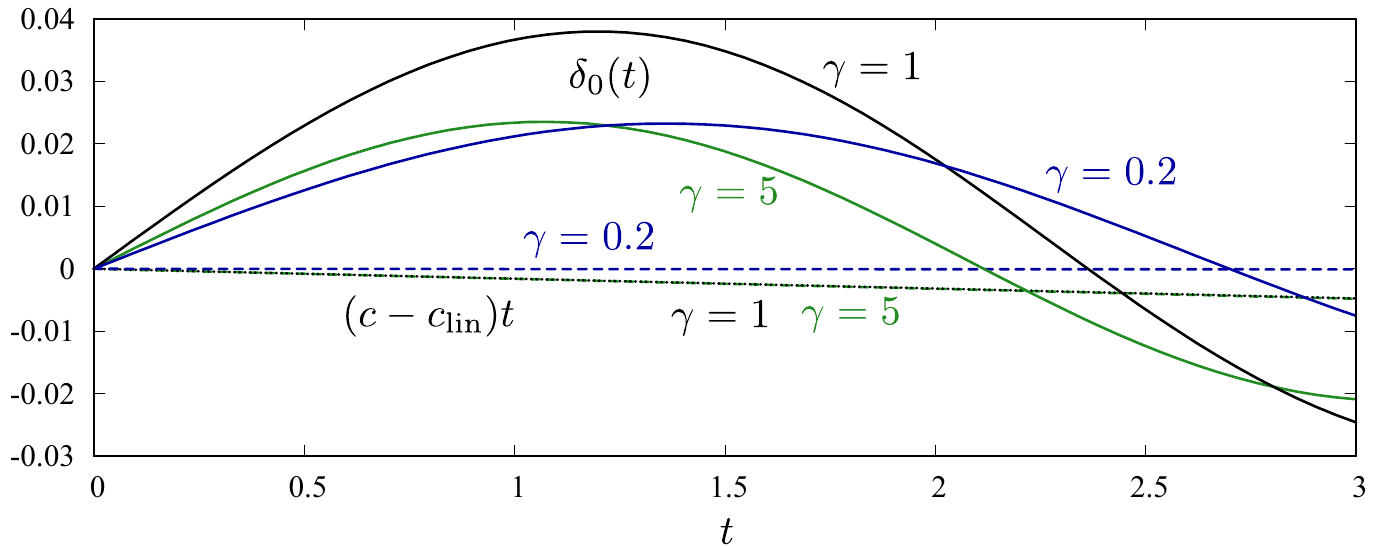}
\caption{\label{delta_plot} Plots of $\delta_0(t)=ct-\alpha_0(t)$ in
  (\ref{eq:exact:a0}) and $(c-c_\text{lin})t$ for the solutions
  of Figure~\ref{traveling_plot}.  }
\end{figure}

Computing the exact solution (\ref{eq:exact:a0}) requires evaluating
$\delta_0(t) = \mc A_0(-ct,-kct)$. We use Newton's method to solve the
implicit equation (\ref{eq:mcA:trav}) for $\mc A_0(x_1,x_2)$ at each
point of a uniform $M\times M$ grid, with $M_1=M_2=M$ in the notation of
Section~\ref{sec:num}. We then use FFTW to compute the 2d Fourier
representation of $\mc A_0(x_1,x_2)$, which is used to quickly evaluate
the function at any point. It would also have been easy to compute
$\mc A_0(-ct,-kct)$ directly by Newton's method, but the Fourier
approach is also very fast and gives more information about the
function $\mc A_0(x_1,x_2)$. In particular, the modes decay to machine
roundoff on the grid, corroborating the assertion in \cite{quasi:ivp}
that $\mc A_0$ is real analytic. We use the exact solution to compute
the error in timestepping (\ref{general_conformal}) and
(\ref{eq:C1:opt2}) from $t=0$ to $t=T$,
\begin{equation*}
  \text{err} = \sqrt{\|\tilde\eta-\tilde\eta_{\text{exact}}\|^2 +
    \|\tilde\varphi-\tilde\varphi_{\text{exact}}\|^2}, \quad
  \|\tilde\eta\|^2 = \frac1{M_1M_2}\sum_{m_1,m_2}
      \tilde\eta\left(\frac{2\pi m_1}{M_1},\frac{2\pi m_2}{M_2},T\right)^2.
\end{equation*}
A detailed convergence study is given in \cite{quasi:ivp} to compare the
accuracy and efficiency of the Runge-Kutta and exponential time
differencing schemes proposed in that paper using the $\gamma=1$
traveling solution above as a test case.  Here we report the errors
for all three waves plotted in Figure~\ref{traveling_plot}
\begin{equation*}
  \begin{array}{c|c|c|c}
    & \gamma=5 & \gamma=1 & \gamma=0.2 \\ \hline
    \text{err} & 1.04\times10^{-16} & 1.16\times10^{-16} & 7.38\times10^{-17}
  \end{array}
\end{equation*}
using the simplest timestepping method proposed in \cite{quasi:ivp} to
solve (\ref{general_conformal}), namely a 5th order explicit
Runge-Kutta method using 900 uniform steps from $t=0$ to $t=3$.  These
errors appear to mostly be due to roundoff error in floating-point
arithmetic, validating the accuracy of both the timestepping algorithm
of \cite{quasi:ivp} and the traveling wave solver of
Section~\ref{sec:num}, which was taken as the exact solution.
Evolving the solutions to compute these errors took less than a second
on a laptop (with $M^2=3600$ gridpoints and 900 timesteps), while
computing the traveling waves via the Levenberg-Marquardt method took
7 seconds on a laptop and only 0.9 seconds on a server (Intel Xeon
  Gold 6136, 3GHz) running on 12 threads (with $M^2=3600$ gridpoints
  and $N_\text{tot}=1200$ unknowns).

\subsection{Larger-amplitude gravity-capillary waves}\label{sec:large:grav:cap}


In the previous sections we studied the full two-parameter family of
quasi-periodic traveling waves with $k=1/\sqrt2$, varying both
$\hat\eta_{1,0}$ and $\hat\eta_{0,1}$ over the range
$[-0.01,0.01]$. Here we search for larger-amplitude waves along the
path $\gamma=1$, where
$\hat\eta_{1,0}=\hat\eta_{0,1}=\hat\eta_\text{max}$ serves as an
amplitude parameter. The calculations are done on an $M\times M$ grid
with Fourier cutoff $N$. As the amplitude increases with $M$ and $N$
fixed, the Fourier modes outside of the cutoff region
$\max(|j_1|,|j_2|)\le N$ eventually grow in magnitude
to exceed $\veps\hat\eta_\text{max}$,
where $\veps$ is machine precision.
Because we formulate the problem as an overdetermined least-squares
problem, it ceases to be possible to satisfy all the equations with
the limited number of Fourier degrees of freedom, and the minimum
value of the objective function begins to grow rapidly with
amplitude.

This is demonstrated in Figure~\ref{fig:errMN} using five
grids ranging from $(M,N)=(48,11)$ to $(M,N)=(240,100)$ and
$\hat\eta_\text{max}$ ranging from $0.001$ to $0.29$. Because the
objective function $f$ is a squared error, if the solution has 14
digits of accuracy the objective function will be around
$10^{-28}$. The coarsest grid becomes under-resolved for
$\hat\eta_\text{max}>0.013$ while the finest grid becomes
under-resolved for $\hat\eta_\text{max}\ge0.0281$.  One can see from
the 2D Fourier plots in Figure~\ref{fourier_plot} that $N=24$ was
overkill at the amplitude $\hat\eta_\text{max}=0.01$ since the modes
have decayed below $\veps\hat\eta_\text{max}=1.11\times10^{-18}$ by
the time $\max(|j_1|,|j_2|)\ge11$. But we see in Figure~\ref{fig:errMN}
that once $\hat\eta_\text{max}$
reaches 0.021, it becomes necessary to increase $M$ and $N$ to
maintain accuracy. At this amplitude, a 2D Fourier plot (not shown)
contains larger-amplitude modes extending all the way to the boundary
of $\max(|j_1|,|j_2|)\le24$.

The running time grows rapidly with grid size, with each calculation
on the grids in Figure~\ref{fig:errMN} requiring an average of
\begin{equation*}
\begin{array}{r|c|c|c|c|c}
  (M,N) & (48,11) & (72,24) & (96,40) & (192,75) & (240,100) \\ \hline
  \text{running time} & 0.1 \text{ sec} & 0.9
  \text{ sec} & 10.2 \text{ sec} & 8.5 \text{ min} & 55.3 \text{ min}
\end{array}
\end{equation*}
on a 3 GHz server with 24 cores.  The memory requirements also grow
rapidly as several matrices of size $M^2\times N_\text{tot}$ are
computed in the Levenberg-Marquardt algorithm, namely the Jacobian and
its (reduced) singular value decomposition. In the $(M,N)=(240,100)$
case, each of these matrices requires 9.3 GB of storage, and we are
not able to increase the problem size further due to hardware
limitations.
As a possible future research direction, one can try to improve the
performance of the Levenberg-Marquardt method for high-dimensional
problems by using a Krylov subspace approximation
\cite{lin2016computationally} without computing the entire Jacobian or
its SVD.

\begin{figure}
\includegraphics[width=.8\textwidth]{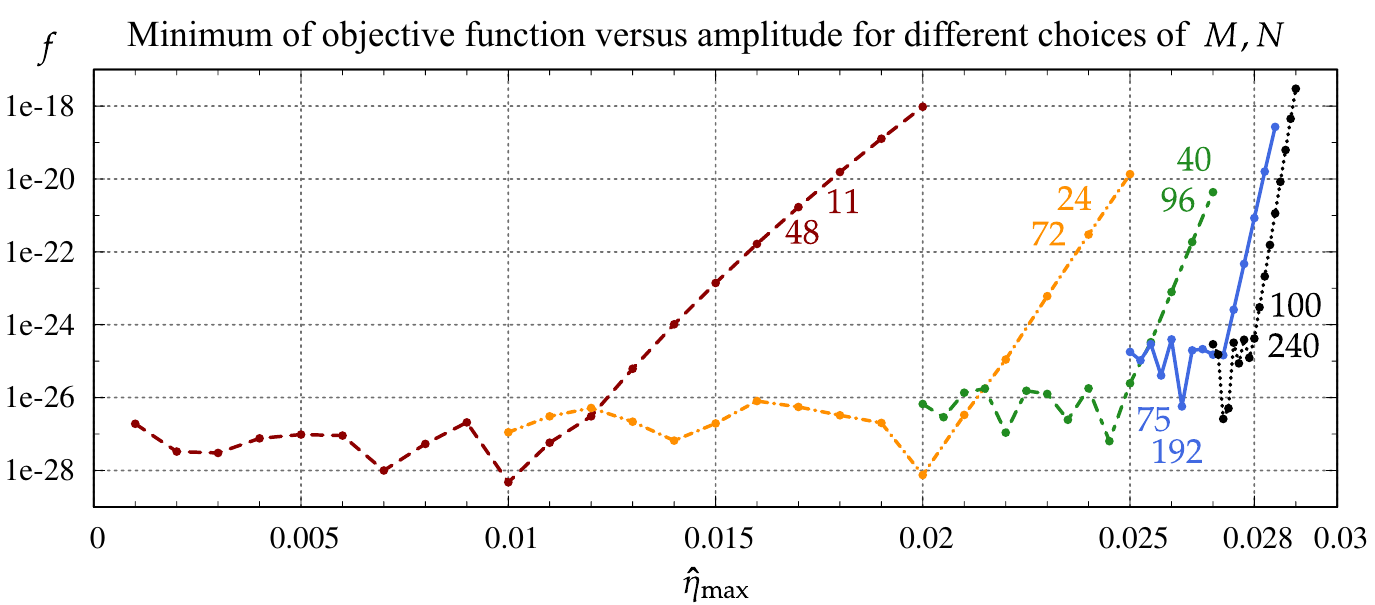}
\caption{\label{fig:errMN}
  Minimum value of the objective function
    $f$ for different values of $M$, $N$ and amplitude,
    $\hat\eta_\text{max}$. Each
    curve is labeled by two numbers, $M$ and $N$, with $N$ the smaller
    one. The objective function grows rapidly with
    $\hat\eta_\text{max}$ once there are not enough Fourier modes to
    represent the solution to machine precision.}
\end{figure}

In Figure~\ref{higher_amplitude}(a), we plot the largest-amplitude
``fully resolved'' solution in Figure~\ref{fig:errMN} with
$(M,N)=(240,100)$ and $(\hat{\eta}_{1, 0}, \hat{\eta}_{0, 1}) =
(0.028, 0.028)$. The solid black curve is the nonlinear traveling
wave, which has a maximum slope of $0.107$ over the representative
interval $[0,10\pi]$ shown in the plot, while the dashed red curve is
the linear prediction $\eta(\alpha) = 0.056\cos(\alpha) +
0.056\cos(\alpha/\sqrt{2})$.  At this amplitude, there is a visible
difference between the nonlinear and linear quasi-periodic waves,
especially near the peaks and troughs of the waves.  However, the
difference is not large since the two base modes are still the
dominant Fourier modes: the amplitudes of the other Fourier modes are
less than 0.0035, which is $1/8$ of the base modes.

One of the main obstacles to computing high-amplitude solutions
numerically is the slow decay of Fourier modes along certain resonant
directions. To demonstrate this, we plot in
Figure~\ref{higher_amplitude}(b) the amplitudes of the Fourier modes
of $\tilde{\eta}$ along 7 directions: $j_1+akj_2\approx 0$ with
$a\in\{1,1.1,0.9,1.5,0.5,0\}$ and $j_2 = 0$. Since $k$ is
irrational, in direction $j_1 + akj_2 \approx 0$ we choose $j_1$ to be
$\min\{\text{floor}(-akj_2), N\}$ with $j_2\in\{-1,\dots,-N\}$.
As shown in the figure, the Fourier modes decay more slowly when the
ratio $-j_1/j_2$ is close to $k$.  Even though $j_1 + k j_2 = 0$ is
the resonant condition for linear quasi-periodic waves, the effects of
this resonance persist into the nonlinear regime.  Along the direction
$j_1+kj_2\approx0$, the mode with $(j_1,j_2)$ farthest from the
origin is $\hat{\eta}_{70, -100}$ and its amplitude is $8.3\times
10^{-12}$, which is the point where floating-point error and finite $N$
truncation effects are roughly equal in this large-scale optimization
problem.

In non-resonant directions, the modes decay faster, often
remaining smaller than $\veps\hat\eta_\text{max}$. For example, curves
(4)--(7) in Figure~\ref{higher_amplitude}(b) drop below $10^{-20}$ for
$(j_1^2+j_2^2)\ge50$, whereas
$\veps\hat\eta_\text{max}=(2^{-53})(0.028)=3.1\times 10^{-18}$. This
may also be observed in the 2D Fourier plots of
Figure~\ref{fourier_plot}.  Presumably the columns $\pa\mc R/\pa p_j$
of the Jacobian corresponding to these modes remain nearly orthogonal
to the residual $\mc R$ throughout the computation, so the
Levenberg-Marquardt algorithm brings them into the calculation with
very small coefficients.  Since increasing the amplitude beyond
$\hat\eta_\text{max}=0.028$ leads to loss of spectral accuracy, this
is the largest amplitude wave of this type that we can compute.

\begin{figure}
\includegraphics[width=\textwidth]{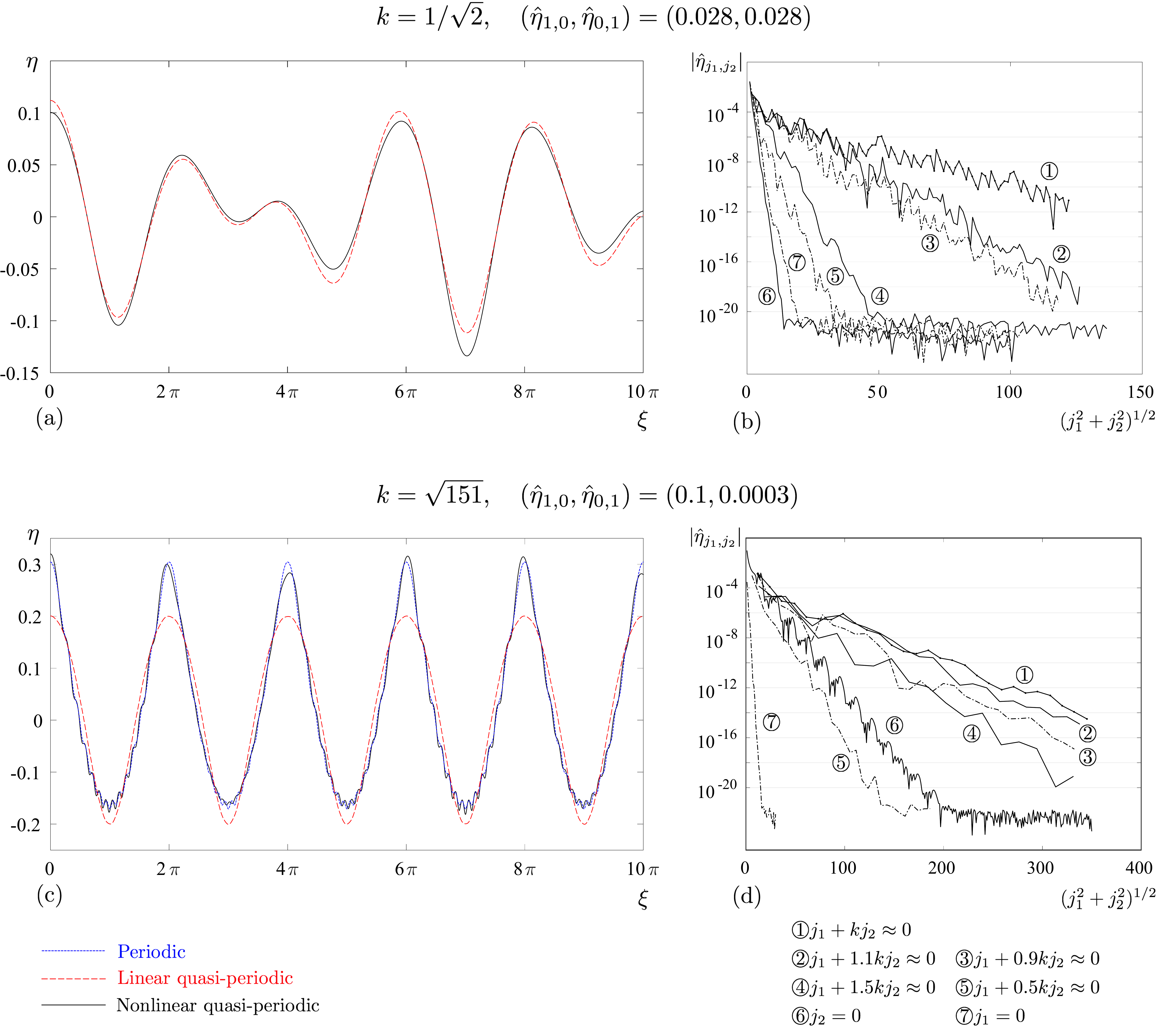}
\caption{\label{higher_amplitude} Plots of higher-amplitude
    quasi-periodic traveling waves.  Panels (a) and (c) show the
    initial conditions $\eta$ over $[0,10\pi]$.  Panels (b) and (d)
    show the amplitudes of Fourier modes along different directions
    versus the magnitude of the two-dimensional mode index
    $(j_1,j_2)$.}
\end{figure}

Next we look for steeper waves by modifying the surface tension
parameter $\tau$ and wave number ratio $k$. So far we have only shown
calculations with $k = 1/\sqrt{2}$, which was an arbitrary choice.
In ocean waves, the characteristic wavelength of gravity waves is
larger than that of capillary waves by several orders of magnitude.
Here we increase $k$ modestly to $\sqrt{151} \approx 12.29$, which is
still much smaller than occurs in the ocean but could be relevant to a
laboratory experiment.  The case of pure gravity waves, which is more
relevant to the ocean, will be undertaken in future
work. Some comments on this were given in the introduction.

Rather than explore the two-parameter family of quasi-periodic water
waves with $k=\sqrt{151}$ near the trivial solution or follow a path
holding $\gamma=\hat\eta_{1,0}/\hat\eta_{0,1}$ constant, we attempt to
compute steep quasi-periodic traveling waves as small quasi-periodic
perturbations of large-amplitude periodic waves, which are
comparatively inexpensive to compute
\cite{dyachenko2016branch,trichtchenko:16}.  In panel (e) of
Figure~\ref{tau_c_plot} above, the contour plot of
$\tau(\hat\eta_{1,0},\hat\eta_{0,1})$ represents a surface in
three-dimensional $(\hat\eta_{1,0},\hat\eta_{0,1},\tau)$ space. The
coordinate planes $\hat\eta_{0,1}=0$ and $\hat\eta_{1,0}=0$ in this
space are two additional surfaces representing traveling waves, the
first of periodic waves of wavelength $2\pi$ and the second of
periodic waves of wavelength $2\pi/k$.  The two parameters on the
$\hat\eta_{0,1}=0$ surface are $\tau$ and $\hat\eta_{1,0}$. This
surface intersects the $\tau(\hat\eta_{1,0},\hat\eta_{0,1})$ surface
along a curve $\tau(\hat\eta_{1,0},0)$ where it is possible to
bifurcate from periodic traveling waves to quasi-periodic traveling
waves. As explained after (\ref{eq:c:cheb:expand}) above,
$\tau(\hat\eta_{1,0},0)$ is an even function of $\hat\eta_{1,0}$, so
its deviation from $\tau_\text{lin}$ is a second-order correction.

In the present case of $k=\sqrt{151}$, we hold the surface tension
fixed at $\tau=\tau_\text{lin}=1/\sqrt{151} \approx 0.0814$ and use
the Levenberg-Marquardt method
\cite{wilkening2012overdetermined,trichtchenko:16} to compute the
resulting one-parameter family of $2\pi$-periodic traveling waves,
denoted as $\eta_\text{per}$, over the range $0\le \hat\eta_{1,0} \le
0.1$.  At the amplitude $\hat\eta_{1,0}=0.1$, the Fourier modes
$\hat\eta_{j_1}$ of $\eta_\text{per}(\alpha)$ decay to machine
precision around $j_1=200$. The maximum slope of this wave in physical
space is 0.332, which is about 3 times steeper than the
quasi-periodic wave computed above with
$k=1/\sqrt2$ and $\hat\eta_{1,0}=\hat\eta_{0,1}=0.028$. Rather than
search within the family of periodic waves for the bifurcation point
$\tau(0.1,0)$, we attempt to jump directly
onto the family of quasi-periodic waves from the periodic wave with
$\tau=\tau_\text{lin}$.  As an initial guess for the
Levenberg-Marquardt method, we set
\begin{equation*}
  \tilde\eta^{(0)}(\alpha_1,\alpha_2) =
  \eta_\text{per}(\alpha_1) + \hat\eta_{0,1}\big(
    e^{i\alpha_2} + e^{-i\alpha_2}\big), \qquad
  \tau^{(0)}=\tau_\text{lin}, \qquad
  b^{(0)} = c_\text{per}^2
\end{equation*}
in (\ref{compute_linear_solution}), where $c_\text{per}$ is the wave
speed of $\eta_\text{per}$. We succeeded in minimizing the objective
function to $f=8.4\times10^{-29}$ holding
$(\hat\eta_{1,0},\hat\eta_{0,1})$ fixed at $(0.1,0.00003)$ and using
$(N_1,N_2)=(216,8)$ for the Fourier cutoffs on an $M_1\times
M_2=576\times24$ grid. Using smaller values $N_2<N_1$ and $M_2<M_1$ is
possible since the unperturbed wave is independent of $\alpha_2$, and
is required to make the problem computationally tractable.  We then use numerical
continuation to increase $\hat\eta_{0,1}$ to $0.0003$ in increments of
$0.00003$, holding $\hat\eta_{1,0}=0.1$ fixed. Polynomial
interpolation of $\tau(0.1,\hat\eta_{0,1})$ from the points
$\hat\eta_{0,1}\in \{\pm 0.00003m\,:\, 1\le m\le 5\}$ gives the value
$\tau(0.1,0)=0.0807311$ for the surface tension of the periodic
traveling wave where the bifurcation to quasi-periodicity occurs. This
is only $0.8\%$ smaller than $\tau_\text{lin}$, which explains why it
was possible to find nearby quasi-periodic waves to the
$\tau=\tau_\text{lin}$ periodic waves even though this is not
the precise location of the bifurcation.

The last solution on this path, with
($\hat\eta_{1,0},\hat\eta_{0,1})=(0.1,0.0003)$, is shown in panels (c)
and (d) of Figure \ref{higher_amplitude}. We had to increase the
Fourier cutoffs $(N_1,N_2)$ to $(350,30)$ and the grid to $M_1\times
M_2 = 720\times 64$ to achieve spectral accuracy. The objective
function for this quasi-periodic solution has been minimized to
$f=1.8\times10^{-25}$ and the maximum slope over the representative
interval $[0,10\pi]$ is $0.448$, so this wave is 35\% steeper than
$\eta_\text{per}$ and 4.2 times steeper than the $k=1/\sqrt2$ wave of
panels (a) and (b) of the figure.  Hardware limitations prevented
increasing $\hat\eta_{0,1}$ further since there are already
$M_1M_2=46080$ nonlinear equations in $N_\text{tot}=21380$ unknowns.
The wave speed and surface tension of this quasi-periodic wave are
$\tau=0.0809677$ and $c=1.072419$, which are close to the values
$\tau(0.1,0)=0.0807311$ and $c(0.1,0)=1.071972$ of the periodic wave
at the bifurcation.

Panel (c) shows the nonlinear periodic and quasi-periodic traveling
waves as well as the linear quasi-periodic traveling wave $\eta = 0.2
\cos(\alpha) + 0.0006 \cos(k\alpha)$ over the representative interval
$[0,10\pi]$. Both nonlinear waves deviate from the linear wave by more
than 50\% of the amplitude of the linear wave, which shows that these
solutions are well outside of the linear regime.  The difference
between the periodic wave and the quasi-periodic wave is also visible,
with the wave peak at $\xi=0$ perturbed upward and the others
perturbed upward or downward and left or right, asymmetrically, in a
non-repeating pattern. The small oscillations in the trough also
change aperiodically from one trough to the next, which shows that
some of the modes $\hat\eta_{j_1,j_2}$ with $j_2\ne0$ are comparable
in size to the modes of the periodic wave responsible for the
capillary ripples in the troughs.

Panel (d) shows the Fourier mode amplitudes $\hat\eta_{j_1,j_2}$ along
various directions in the $(j_1,j_2)$ lattice. Along the direction
$j_1 + akj_2 \approx 0$ we choose $j_1$ to be
$\min\{\text{floor}(-akj_2), N_1\}$ with $j_2\in\{-1,\dots,-N_2\}$.
One can see that the Fourier modes decay more slowly along directions
$j_1 + akj_2 \approx 0$ when $a\in\{1,1.1,0.9,1.5\}$ than when
$a\in\{0.5,0\}$ or when $j_2\approx0$.  Thus, the linear resonance
condition $j_1 + k j_2 = 0$ continues to have a large effect on the
Fourier modes in the nonlinear regime. As noted in
Section~\ref{sec:trav_rslts} above, we believe this is because the
corresponding modes $e^{i(j_1+j_2k)\alpha}$ in the expansion of
$\eta(\alpha)$ in (\ref{eq:eta:tilde}) have long wavelengths and are
not as strongly controlled by the governing equations
(\ref{eq:trav:conf}) as other modes, which leads to greater
sensitivity to nonlinear interactions among the Fourier modes.

\color{black}


\section{Conclusion} 
\label{sec:conclude}


In this work, we have formulated the two-dimensional, infinite depth
gravity-capillary traveling wave problem in a spatially
quasi-periodic, conformal mapping framework. We have numerically
demonstrated the existence of traveling solutions that are a
quasi-periodic generalizations of Wilton's ripples. To compute them,
we adapted an overdetermined nonlinear least squares technique
introduced in \cite{wilkening2012overdetermined} for a different
problem. For each solution computed, the value of $k$ and the
amplitudes of two base Fourier modes $\hat{\eta}_{1,0}$ and
$\hat{\eta}_{0,1}$ are fixed while $\tau$, $c$ and the other Fourier
modes $\hat\eta_{j_1,j_2}$ are varied to search for solutions of
(\ref{eq:trav:conf}). Before minimizing (\ref{numerical_objective}),
the initial guess for each solution is computed using either the
linear approximation (\ref{linear_solution}) or numerical
continuation.  We obtained quasi-periodic traveling solutions
  with maximum slope as large as 0.448 and validated the accuracy of
the traveling solutions using the timestepping algorithm of
\cite{quasi:ivp}. To evolve at constant speed in physical space, we
have shown that the 2d representation of the quasi-periodic waves
travel at a nonuniform speed through the torus. We explain this
  by constructing a change of variables, namely (\ref{eq:mcA:meaning}),
  relating quasi-periodic functions in conformal space to
  quasi-periodic functions in physical space with the same wave number
  ratio $k$.

As the amplitude increases, we have found that the wave spectrum
  of a quasi-periodic traveling wave continues to decay exponentially,
  but becomes much broader than in the linear and weakly nonlinear
  approximations. For example, the solution shown in panels (c) and
  (d) of Figure~\ref{higher_amplitude} has 23 modes within one percent
  of $\hat\eta_{1,0}$, 265 within a factor of $10^{-5}$, 1500 within
  an factor of $10^{-9}$, and 21380 involved in the calculation. We
  also demonstrated the nonlinear dependence of wave speed, surface
  tension, energy and momentum for the two-parameter family with
  amplitude parameters in the range $\max\{\big|\hat\eta_{1,0}\big| ,
  \big|\hat\eta_{0,1}\big|\}\le0.01$.  Resonance effects were always
  observed in the Fourier modes $\hat\eta_{j_1,j_2}$ near the line
  $j_1+j_2k=0$, which is the resonance condition for linear waves.  We
  provided the explanation that these modes are slowly varying when
  evaluated along the characteristic direction $(1,k)$ in the torus,
  and therefore are not strongly controlled by the Euler equations
  even for large-amplitude waves in the nonlinear regime. Additional
  resonance effects could be investigated in the future using
  Fourier-Bloch stability techniques \cite{longuet:78, oliveras:11,
    trichtchenko:16} generalized to the case of
  large-amplitude quasi-periodic traveling waves.

The question of what happens in our framework if $k$ is rational is
interesting.  We believe the initial value problem
(\ref{general_conformal}) could still be solved, though in that case
solving the torus version of the equations is equivalent to
simultaneously computing a family of 1d solutions on a periodic
domain. Families of 1d waves corresponding to a single solution of the
torus problem are discussed in detail in \cite{quasi:ivp}, and take
the form (\ref{eq:family:full}) above. If $k=q/p$ with $p$ and $q$
relatively prime integers, the waves in this family all have period
$2\pi p$.  The traveling wave problem becomes degenerate if $k$ is
rational --- solutions of the torus version of (\ref{eq:trav:conf})
may still exist (we do not know), but if so, they are not
unique. Indeed, if $k=q/p$ as above and $\tilde\eta_1$ solves the
torus version of (\ref{eq:trav:conf}), then for any $2\pi$-periodic,
real analytic function $\alpha_0(r)$,
\begin{equation}
  \tilde\eta_2\begin{pmatrix} \alpha_1\\ \alpha_2\end{pmatrix} =
  \tilde\eta_1\left(\begin{pmatrix} \alpha_1\\ \alpha_2\end{pmatrix} -
      \begin{pmatrix} p \\ q \end{pmatrix}\,\alpha_0
      \big( -q\alpha_1 + p\alpha_2 \big)\right)
\end{equation}
will also be a solution of (\ref{eq:trav:conf}) since the
corresponding 1d functions passing through the torus along
characteristic lines are related by a simple reparametrization,
\begin{equation}
  \eta_2(\alpha;\theta) = \tilde\eta_2\begin{pmatrix} \alpha \\ \theta + k\alpha
  \end{pmatrix} = \tilde\eta_1\begin{pmatrix} \alpha-p\alpha_0(p\theta) \\ 
  \theta + k\alpha - q\alpha_0(p\theta) \end{pmatrix} = 
  \eta_1\big(\alpha-p\alpha_0(p\theta);\theta\big).
\end{equation}
Another degeneracy is that the modes $\hat\eta_{j_1,j_2}$ of a
solution of (\ref{eq:trav:conf}) with $j_1+kj_2=0$ and
$(j_1,j_2)\ne(0,0)$ can be modified arbitrarily (maintaining
  $\hat\eta_{-j_1,-j_2}=\overline{\hat\eta_{j_1,j_2}}$) to obtain
additional solutions of (\ref{eq:trav:conf}). These modes are plane
waves that only affect the 1d functions passing through the torus
along characteristic lines by an additive constant. The resonance
phenomenon observed in the Fourier modes in Figure \ref{fourier_plot}
is presumably a small-divisor phenomenon \cite{rard2009small} in the
irrational case related to this degeneracy.  If solutions for rational
$k$ exist, a natural open question is whether they can be selected to
fit together continuously with solutions for nearby irrational wave
numbers.  In floating point arithmetic, irrational wave numbers are
approximated by rational ones. We did not encounter difficulties with
this, presumably because the above degeneracies are not visible with
the grid resolution used. More work is needed to understand this
rigorously.

The amplitude ratio $\gamma=\hat\eta_{1,0}/\hat\eta_{0,1}$ plays an
important role in determining the shapes of smaller-amplitude
solutions. As seen in Figures~\ref{initial_plot}
and~\ref{traveling_plot}, the quasi-periodic features of the solutions
are most evident when $\gamma \approx 1$.  For larger-amplitude
  waves such as the one plotted in Figure~\ref{higher_amplitude}(c),
  quasi-periodicity can lead to visible changes from one peak or
  trough to the next, without ever repeating. It is remarkable that
  such a complicated aperiodic wavetrain is a stationary solution of
  the Euler equations in a moving frame.

In the future, we hope to explore the long-time dynamics of
  unstable subharmonic perturbations of periodic waves; to search for
  quasi-periodic traveling waves that bifurcate from large-amplitude
  periodic gravity waves or from overhanging gravity-capillary waves;
  to study the behavior of different perturbation families,
e.g.~fixing the amplitudes of different base Fourier modes in
(\ref{resonance_fourier}) such as $\hat{\eta}_{1,0}$ and
$\hat{\eta}_{1,1}$; to study the stability of co-propagating
  quasi-periodic traveling waves and compare to the effects of oblique
  multi-phase interacting wave trains \cite{onorato2006modulational,
    ablowitz2015interacting}; to develop a generalization of
  Fourier-Bloch stability analysis for quasi-periodic waves, which
  presumably will further increase the number of quasi-periods of the
  perturbed wave; and to study finite-depth effects on both the
  initial value problem and the traveling wave problem. Additional
  future research challenges include establishing rigorous existence
  proofs; improving the algorithm to employ a Newton-Krylov or
  limited-memory approach so that it is not necessary to compute or
  factor the entire Jacobian matrix; and developing a formulation for
  quasi-periodic three-dimensional water waves, which would require
  abandoning the conformal mapping framework.

Declaration of Interests. The authors report no conflict of interest.

\appendix

\section{Dynamics of Traveling Waves in Conformal Space}
\label{sec:dyn:trav}

In this section we study the dynamics of the traveling waves of
Section~\ref{sec:gov:trav} under the evolution equations
(\ref{general_conformal}) for various choices of $C_1$. We show that
the waves maintain a permanent form but generally travel at a
non-uniform speed in conformal space.  We start by showing that there
is a choice of $C_1$ for which $\eta$ and $\varphi$ remain stationary
in time. We then show how $C_1$ changes when the waves are phase
shifted by $\alpha_0(t)$, and how to determine $\alpha_0(t)$ so that
$C_1$ takes the value in (\ref{eq:C1:opt2}).  The evolution of the
torus version of (\ref{eq:trav:conf}) under (\ref{general_conformal})
is also worked out.

We will need the following theorem and corollary, proved in
\cite{quasi:ivp}:

\begin{theorem}\label{thm:conformal}
  Suppose $\veps>0$ and $z(w)$ is analytic on the half-plane $\mbb
  C^-_\veps = \{w\;:\; \im w<\veps\}$. Suppose there is a constant
  $M>0$ such that $|z(w)-w|\le M$ for $w\in\mbb C^-_\veps$, and that the
  restriction $\zeta=z\vert_\mbb R$ is injective. Then the curve
  $\zeta(\alpha)$ separates the complex plane into two regions, and
  $z(w)$ is an analytic isomorphism of the lower half-plane onto the
  region below the curve $\zeta(\alpha)$.
\end{theorem}

\begin{cor}\label{cor:conformal}
  Suppose $k>0$ is irrational,
  $\tilde\eta(\alpha_1,\alpha_2)=\sum_{(j_1,j_2)\in\mbb Z^2}
  \hat\eta_{j_1,j_2}e^{i(j_1\alpha_1+j_2\alpha_2)}$, and there exist constants $C$
          and $\veps>0$ such that
  \begin{equation}
    \hat\eta_{-j_1,-j_2}=\overline{\hat\eta_{j_1,j_2}}, \qquad
    \big|\hat\eta_{j_1,j_2}\big|\le Ce^{-3\veps K\max(|j_1|,|j_2|)},
    \qquad\quad
    (j_1,j_2)\in\mbb Z^2,
  \end{equation}
  where $K=\max(k,1)$.  Let $x_0$ be real and define
  $\tilde\xi=x_0+H[\tilde\eta]$, $\tilde\zeta= \tilde\xi+i\tilde\eta$
  and
  \begin{equation}
    \tilde z(\alpha_1,\alpha_2,\beta) = x_0 + i\hat\eta_{0,0} +
    \sum_{j_1+j_2k<0} 2i\hat\eta_{j_1,j_2}e^{-(j_1+j_2k)\beta}e^{i(j_1\alpha_1+
        j_2\alpha_2)}, \qquad (\beta<\veps),
  \end{equation}
  where the sum is over all integer pairs $(j_1,j_2)$ satisfying the
  inequality.
  Suppose also that for each fixed $\theta\in[0,2\pi)$, the function
  $\alpha\mapsto\zeta(\alpha;\theta)= \alpha+\tilde
  \zeta(\alpha,\theta+k\alpha)$ is injective from $\mbb R$ to $\mbb C$
  and $\zeta_\alpha(\alpha;\theta)\ne0$ for $\alpha\in\mbb R$. Then
  for each $\theta\in\mbb R$, the curve $\zeta(\alpha;\theta)$
  separates the complex plane into two regions and
  \begin{equation}
    z(\alpha+i\beta;\theta) = (\alpha+i\beta) +
    \tilde z(\alpha,\theta+k\alpha,\beta), \qquad
    (\beta<\veps)
  \end{equation}
  is an analytic isomorphism of the lower half-plane onto the
  region below $\zeta(\alpha;\theta)$. Moreover, there is a constant
  $\delta>0$ such that $|z_w(w;\theta)|\ge\delta$ for
  $\im w\le 0$ and $\theta\in\mbb R$.
\end{cor}

We now prove a theorem and two corollaries that describe the
dynamics of traveling waves in conformal space under the evolution
equations (\ref{general_conformal}) for various choices of $C_1$.

\begin{theorem}\label{thm:trav:C1}
  Suppose $\tilde\eta_0(\alpha_1,\alpha_2)$ satisfies the
  torus version of (\ref{eq:trav:conf}) as well as the
  assumptions in Corollary~\ref{cor:conformal}. Define $\tilde\xi_0 =
  \Hilbert[\tilde\eta_0]$, $\tilde\zeta_0=\tilde\xi_0+i\tilde\eta_0$
  and $\tilde\varphi_0=c\tilde\xi_0$.  Let
  $\eta_0(\alpha;\theta)=\tilde\eta_0(\alpha,\theta+k\alpha)$,
  $\varphi_0(\alpha;\theta)=\tilde\varphi_0(\alpha,\theta+k\alpha)$,
  $\xi_0(\alpha;\theta)=\alpha+\tilde\xi_0(\alpha,\theta+k\alpha)$ and
  $\zeta_0=\xi_0+i\eta_0$.  Suppose that for each $\theta\in[0,2\pi)$,
  $\alpha\mapsto\zeta_0(\alpha;\theta)$ is injective, i.e.~none of the
  curves in the family (\ref{eq:trav:fam:eta}) self-intersect.  Then
  for each $\theta\in\mbb R$,
\begin{equation}\label{eq:zeta:phi:trav1}
  \zeta(\alpha,t;\theta) = \zeta_0(\alpha;\theta)+ct, \qquad
  \varphi(\alpha,t;\theta) = \varphi_0(\alpha;\theta)
\end{equation}
satisfy (\ref{general_conformal}) with $C_1=cP_0[\xi_\alpha/J]$.
\end{theorem}

\begin{proof}
  We have assumed the initial reconstruction of $\xi$ from $\eta$
  yields $\xi(\alpha,0;\theta)=\xi_0(\alpha;\theta)$, so $x_0(0)=0$ in
  (\ref{eq:xi:from:eta}). We need to show that $\eta_t=0$,
  $\varphi_t=0$ and $dx_0/dt=c$ in (\ref{general_conformal}), from
  which it follows that
  $\xi(\alpha,t;\theta)=\xi_0(\alpha;\theta)+ct$. Since
  $\tilde\xi_0=\Hilbert[\tilde\eta_0]$ and none of the curves in the
  family (\ref{eq:trav:fam:eta}) self-intersect,
  Theorem~\ref{thm:conformal} and Corollary~\ref{cor:conformal} above
  show that the holomorphic extension from $\zeta_0(\alpha;\theta)$ to
  $z_0(w;\theta)$ is an analytic isomorphism of the lower half-plane
  to the fluid region, and $1/|z_{0,w}|$ is uniformly bounded. In
  (\ref{general_conformal}), we define
  $\xi_\alpha=1+\Hilbert[\eta_\alpha]$, $\psi=-\Hilbert[\varphi]$,
  $J=\xi_\alpha^2+\eta_\alpha^2$ and $\chi=\psi_\alpha/J$. This
  formula for $\xi_\alpha$ gives the same result as differentiating
  $\xi(\alpha,t;\theta)$ in (\ref{eq:zeta:phi:trav1}) with respect to
  $\alpha$. From $\tilde\varphi_0=c\tilde\xi_0$ and
  $\hat\eta_{0,0}=0$, we have $\chi=c\eta_\alpha/J$.  The extension of
  $\zeta(\alpha,t;\theta)$ to the lower half-plane is
  $z(w,t;\theta)=[z_0(w;\theta)+ct]$.  We have not yet established
  that $\zeta(\alpha,t;\theta)$ solves (\ref{general_conformal}), but
  we know $z_t/z_w$ is bounded in the lower half-plane, so there is a
  $C_1$ such that
  \begin{equation}\label{eq:Hchi:C1}
    \begin{pmatrix}
      -H\chi + C_1 \\  -\chi
    \end{pmatrix} =
    \frac1J\begin{pmatrix}
    \xi_\alpha & \eta_\alpha \\
    -\eta_\alpha & \xi_\alpha
    \end{pmatrix}\begin{pmatrix}
      c \\ 0 \end{pmatrix},
  \end{equation}
  where the right-hand side represents complex division of $z_t$ by
  $z_\alpha$. Since $P_0\Hilbert\chi=0$, we learn from
  (\ref{eq:Hchi:C1}) that $C_1=cP_0[\xi_\alpha/J]$.  But
  $\xi_t$ and $\eta_t$ in (\ref{eq:xi:t:eta:t}) are obtained
  by multiplying (\ref{eq:Hchi:C1}) by
  $[\xi_\alpha,-\eta_\alpha;\eta_\alpha, \xi_\alpha]$, which gives
  $\xi_t=c$, $\eta_t=0$. Equation (\ref{eq:x0:evol}) is then
  $dx_0/dt=P_0[\xi_t]=c$. Finally, using $\chi=c\eta_\alpha/J$,
  $H\chi = C_1-c\xi_\alpha/J$, $\varphi_\alpha=c(\xi_\alpha-1)$ and
  $\psi_\alpha=c\eta_\alpha$ in (\ref{general_conformal}) gives
  \begin{equation}
    \begin{aligned}
      \varphi_t &= P\bigg[\frac{\psi_\alpha^2 - \varphi_\alpha^2}{2J} -
    \varphi_\alpha\Hilbert[\chi] + C_1\varphi_\alpha - g\eta +
    \tau\kappa\bigg] \\
      &= P\bigg[
        \frac{c^2\eta_\alpha^2 - c^2(\xi_\alpha^2-2\xi_\alpha+1)}{2J} +
        c\frac{c(\xi_\alpha-1)\xi_\alpha}{J} - g\eta + \tau\kappa\bigg] \\
      &= P\bigg[\frac{c^2}{2J}\Big(J-1\Big) - g\eta + \tau\kappa\bigg]
      = P\bigg[ -\frac{c^2}{2J}-g\eta + \tau\kappa\bigg] = 0,
    \end{aligned}
  \end{equation}
  where we used (\ref{eq:trav:conf}) in the last step.
\end{proof}

\begin{cor}
  Suppose $\tilde\zeta_0(\alpha_1,\alpha_2)$,
  $\tilde\varphi_0(\alpha_1,\alpha_2)$, $\zeta_0(\alpha;\theta)$ and
  $\varphi_0(\alpha;\theta)$ satisfy the hypotheses of
  Theorem~\ref{thm:trav:C1} and $\alpha_0(t)$ is any continuously
  differentiable, real-valued function.  Then
  \begin{equation}\label{eq:zeta:phi:a0}
    \zeta(\alpha,t;\theta) = \zeta_0(\alpha-\alpha_0(t);\theta)
    +ct, \qquad
    \varphi(\alpha,t; \theta) = \varphi_0(\alpha-\alpha_0(t);\theta)
  \end{equation}
  are solutions of (\ref{general_conformal}) with
  $C_1=cP_0[\xi_\alpha/J]-\alpha_0'(t)$. The corresponding
  solutions of the torus version of (\ref{general_conformal}) for
  this choice of $C_1$ are
  \begin{equation}\label{eq:zeta:phi:a0:torus}
    \begin{aligned}
      \tilde\zeta(\alpha_1,\alpha_2,t) &=
      \tilde\zeta_0\big(\alpha_1-\alpha_0(t),
        \alpha_2-k\alpha_0(t)\big) + ct - \alpha_0(t), \\
      \tilde\varphi(\alpha_1,\alpha_2,t) &=
      \tilde\varphi_0\big(\alpha_1-\alpha_0(t),
        \alpha_2-k\alpha_0(t)\big).
    \end{aligned}
  \end{equation}
\end{cor}

\begin{proof}
  Since $\pa_\alpha$ and $\Hilbert$ commute with
  $\alpha$-translations, substitution of
  $\eta_0(\alpha-\alpha_0(t);\theta)$ and
  $\varphi_0(\alpha-\alpha_0(t);\theta)$ in the right-hand sides of
  (\ref{general_conformal}) without changing $C_1$ would still lead to
  $\eta_t=0$, $\varphi_t=0$ and $dx_0/dt=c$, and
  (\ref{eq:xi:t:eta:t}) would still give $\xi_t=c$. Including
  $-\alpha_0'(t)$ in $C_1$ leads instead to
  $\eta_t=-\alpha_0'(t)\eta_\alpha$ and
  $\varphi_t=-\alpha_0'(t)\varphi_\alpha$ in (\ref{general_conformal})
  and $\xi_t=c-\alpha_0'(t)\xi_\alpha$ in
  (\ref{eq:xi:t:eta:t}), which are satisfied by
  (\ref{eq:zeta:phi:a0}). It also leads to $dx_0/dt=[c-\alpha_0'(t)]$
  in (\ref{eq:x0:evol}), which keeps the reconstruction of $\xi$ from
  $\eta$ via (\ref{eq:xi:from:eta}) consistent with the evolution
  equation for $\xi_t$.
  
  The functions in (\ref{eq:zeta:phi:a0}) and
  (\ref{eq:zeta:phi:a0:torus}) are related by
  \begin{equation}\label{eq:zeta:phi:a0:2}
    \zeta(\alpha,t;\theta)=\alpha+\tilde\zeta(\alpha,\theta+k\alpha,t), \qquad
    \varphi(\alpha,t;\theta)=
    \tilde\varphi(\alpha,\theta+k\alpha,t).
  \end{equation}
  Applying the 1d version of (\ref{general_conformal}) to
  (\ref{eq:zeta:phi:a0:2}) is equivalent to applying the torus version
  of (\ref{general_conformal}) to (\ref{eq:zeta:phi:a0:torus}) and
  evaluating at $(\alpha,\theta+k\alpha,t)$. Since
  (\ref{eq:zeta:phi:a0}) satisfies the 1d version of
  (\ref{general_conformal}) and every point
  $(\alpha_1,\alpha_2)\in\mbb T^2$ can be written as
  $(\alpha,\theta+k\alpha)$ for some $\alpha$ and $\theta$,
  (\ref{eq:zeta:phi:a0:torus}) satisfies the torus version
  of (\ref{general_conformal}).
\end{proof}

\begin{cor}\label{cor:trav:conf}
  Suppose $\tilde\zeta_0(\alpha_1,\alpha_2)$,
  $\tilde\varphi_0(\alpha_1,\alpha_2)$, $\zeta_0(\alpha;\theta)$ and
  $\varphi_0(\alpha;\theta)$ satisfy the hypotheses of
  Theorem~\ref{thm:trav:C1} and $\xi_{0,\alpha}(\alpha;\theta)>0$ for
  $\alpha\in[0,2\pi)$ and $\theta\in[0,2\pi)$. Then if $C_1$ is
  chosen as in (\ref{eq:C1:opt2}) to maintain $\tilde\xi(0,0,t)=0$,
  the solution of the torus version of (\ref{general_conformal}) with
  initial conditions
  \begin{equation}\label{eq:init:trav}
    \tilde\zeta(\alpha_1,\alpha_2,0) = \tilde\zeta_0(\alpha_1,\alpha_2), \qquad
    \tilde\varphi(\alpha_1,\alpha_2,0) = \tilde\varphi_0(\alpha_1,\alpha_2)
  \end{equation}
  has the form (\ref{eq:zeta:phi:a0:torus}) with
  \begin{equation}\label{eq:a0:ct}
    \alpha_0(t)=ct-\mc A_0(-ct,-kct),
  \end{equation}
  where $\mc A_0(x_1,x_2)$ is defined implicitly by
  \begin{equation}\label{eq:mcA:trav}
    \mc A_0(x_1,x_2) + \tilde\xi_0\big(x_1+\mc A_0(x_1,x_2)\,,\,
    x_2+k\mc A_0(x_1,x_2)\big) = 0, \qquad
      (x_1,x_2)\in\mbb T^2.
  \end{equation}
\end{cor}

\begin{proof}
  The assumption that $\xi_{0,\alpha}(\alpha;\theta)>0$
  ensures that all the waves in the family $\zeta_0(\alpha;\theta)$
  are single-valued and have no vertical tangent lines. Under these
  hypotheses, it is proved in \cite{quasi:ivp} that there is a unique
  function $\mc A_0(x_1,x_2)$ satisfying (\ref{eq:mcA:trav}) and that it
  is real analytic and periodic. We seek a solution of the form
  (\ref{eq:zeta:phi:a0:torus}) satisfying $\tilde\xi(0,0,t)=0$,
  \begin{equation}
    \begin{aligned}
      & \tilde\xi(0,0,t) = 
     \tilde\xi_0(-\alpha_0(t),-k\alpha_0(t)) +ct - \alpha_0(t) \\
    & \qquad = [ct-\alpha_0(t)] + \tilde\xi_0\big(-ct+[ct-\alpha_0(t)],
      -kct+k[ct-\alpha_0(t)]\big) = 0.
    \end{aligned}
  \end{equation}
  Comparing with (\ref{eq:mcA:trav}), we find that
  $[ct-\alpha_0(t)]=\mc A_0(-ct,-kct)$, which is (\ref{eq:a0:ct}).
  Since $\tilde\eta_0(\alpha_1,\alpha_2)$ is even,
  $\tilde\xi_0=\Hilbert[\tilde\eta_0]$ is odd and $\mc A_0(0,0)=0$.
  Thus, $\alpha_0(0)=0$ and the initial conditions (\ref{eq:init:trav})
  are satisfied. Since $\xi(0,0,t)=0$, $C_1$ satisfies
  (\ref{eq:C1:opt2}).
\end{proof}


\bibliographystyle{abbrv}

\end{document}